\documentclass[aps,pra,reprint,amsmath,superscriptaddress,amssymb,showpacs,nofootinbib]{revtex4-2}
\usepackage{amsmath,braket,amssymb}
\usepackage[final]{graphicx}
\usepackage{subfigure}
\usepackage[english]{babel}
\usepackage[utf8]{inputenc}
\usepackage{mathtools}
\usepackage{empheq}
\usepackage{tensor}
\usepackage{cancel}
\usepackage{enumitem}
\usepackage{simplewick}

\usepackage[usenames,dvipsnames]{color}

\usepackage{enumitem}
\usepackage{slashed}
\usepackage{graphicx}
\usepackage{xcolor}

\colorlet{Green}{black!30!green}

\usepackage{tikz}
\usetikzlibrary{calc,fadings,decorations.pathreplacing,shapes,shapes.multipart,arrows,shapes.misc,intersections,positioning,patterns}
\tikzset{arrow data/.style 2 args={%
		decoration={%
			markings,
			mark=at position #1 with \arrow{#2}},
		postaction=decorate}
}
\usepackage{bm}
\usepackage[colorlinks=true,citecolor=blue,linkcolor=blue,urlcolor=blue]{hyperref}
\usepackage{dsfont}
\usepackage{changepage}
\usepackage{array}
\usepackage{soul}
\usepackage[capitalise]{cleveref}
\crefname{section}{Sec.}{Secs.}
\Crefname{section}{Sec.}{Secs.}
\usepackage{xifthen}
\usepackage{xargs}
\usepackage{dynkin-diagrams}

\usepackage{amsthm}
\theoremstyle{definition}

\theoremstyle{plain}

\newtheorem{thm}{Theorem}

\newtheorem{lem}{Lemma}

\usepackage{bm}

\graphicspath{{./}{./figures/}}

\newcommand{\bit}{\begin{itemize}}
	\newcommand{\eit}{\end{itemize}}

\renewcommand{\>}{\right\rangle}
\newcommand{\<}{\left\langle}
\newcommand{\ba}{\begin{align}}
	\newcommand{\ea}{\end{align}}
\newcommand{\be}{\begin{equation}}
	\newcommand{\ee}{\end{equation}}
\newcommand{\bi}{\begin{itemize}}
	\newcommand{\ei}{\end{itemize}}

\newcommand{\Tr}{\operatorname{Tr}}

\DeclareMathAlphabet{\mymathbb}{U}{BOONDOX-ds}{m}{n}

\renewcommand{\log}{\ln}

\usepackage{braket}
\usepackage[normalem]{ulem}

\begin{document}
	\date{\today}

	\newcommand{\bbra}[1]{\<\< #1 \right|\right.}
	\newcommand{\kket}[1]{\left.\left| #1 \>\>}
	\newcommand{\bbrakket}[1]{\< \Braket{#1} \>}
	\newcommand{\pll}{\parallel}
	\newcommand{\nn}{\nonumber}
	\newcommand{\transp}{\text{transp.}}
	\newcommand{\nor}{z_{J,H}}
	
	\newcommand{\hL}{\hat{L}}
	\newcommand{\hR}{\hat{R}}
	\newcommand{\hQ}{\hat{Q}}

	\title{Many-body entropies and entanglement from polynomially-many local measurements}
	
	\begin{abstract}
		Estimating global properties of many-body quantum systems such as entropy or bipartite entanglement is a notoriously difficult task, typically requiring a number of measurements or classical post-processing resources growing exponentially in the system size. In this work,
		we address the problem of estimating global entropies and mixed-state entanglement via partial-
		transposed (PT) moments, and show that efficient estimation strategies exist under the assumption
		that all the spatial correlation lengths are finite. Focusing on one-dimensional systems, we identify a
		set of approximate factorization conditions (AFCs) on the system density matrix which allow us to
		reconstruct entropies and PT moments from information on local subsystems. This yields a simple and efficient strategy for entropy and entanglement estimation. Our method could be implemented in different ways, depending on how information on local subsystems is extracted. Focusing on randomized measurements (RMs), providing a practical and common measurement scheme, we prove that our protocol only requires
			polynomially-many measurements and post-processing operations, assuming that the state to be measured satisfies the AFCs. We prove that the AFCs hold
		for finite-depth quantum-circuit states and translation-invariant matrix-product density operators,
		and provide numerical evidence that they are satisfied in more general, physically-interesting cases,
		including thermal states of local Hamiltonians. We argue that our method could be practically
		useful to detect bipartite mixed-state entanglement for large numbers of qubits available in today’s
		quantum platforms.
	\end{abstract}
	
	\author{Benoît Vermersch}
	\affiliation{Univ. Grenoble Alpes, CNRS, LPMMC, 38000 Grenoble, France}
	\affiliation{Institute for Theoretical Physics, University of Innsbruck, 6020 Innsbruck, Austria}
	\affiliation{Institute for Quantum Optics and Quantum Information of the Austrian Academy of Sciences, 6020 Innsbruck, Austria}
	
	\author{Marko Ljubotina}
	\affiliation{Institute of Science and Technology Austria (ISTA), Am Campus 1, 3400 Klosterneuburg, Austria}
	
	\author{J. Ignacio Cirac}
	\affiliation{Max-Planck-Institut für Quantenoptik, Hans-Kopfermann-Straße 1, D-85748 Garching, Germany}
	\affiliation{Munich Center for Quantum Science and Technology (MCQST), Schellingstraße 4, D-80799 M\"unchen, Germany}
	
	\author{Peter Zoller}
	\affiliation{Institute for Theoretical Physics, University of Innsbruck, 6020 Innsbruck, Austria}
	\affiliation{Institute for Quantum Optics and Quantum Information of the Austrian Academy of Sciences, 6020 Innsbruck, Austria}
	
	\author{Maksym Serbyn}
	\affiliation{Institute of Science and Technology Austria (ISTA), Am Campus 1, 3400 Klosterneuburg, Austria}
	
	\author{Lorenzo Piroli}
	\affiliation{Dipartimento di Fisica e Astronomia, Università di Bologna and INFN, Sezione di Bologna, via Irnerio 46, I-40126 Bologna, Italy}
	
	\maketitle
	
	\tableofcontents
	
	\section{Introduction}
	
	In the context of today's digital quantum technologies~\cite{blatt2012quantum,gross2017quantum,schafer2020tools,kjaergaard2020superconducting,morgado2021quantum,monroe2021programmable,alexeev2021quantum,pelucchi2022potential,burkard2023semiconductor}, an outstanding challenge is to devise measurement schemes of many-qubit states which are efficient and yet simple enough to be performed in current noisy intermediate-scale quantum (NISQ) devices~\cite{preskill2018quantum,altman2021quantum}. This problem has motivated new ideas and protocols to improve our ability to characterize complex quantum states. A notable example is that of the so-called randomized-measurement (RM) toolbox~\cite{elben2023randomized,elben2019statistsical,cieslinski2023analysing,huang2020predicting}, which has provided us with novel opportunities to experimentally investigate entanglement, a cornerstone in both quantum-information~\cite{horodecki2009quantum,nielsen2010quantum} and quantum many-body theory~\cite{calabrese2009entanglement,eisert2010colloquium}.
	
	For instance, paralleling earlier experiments based on quantum interference~\cite{islam2015measuring,kaufman2016quantum,linke2018measuring}, cf.~also Refs.~\cite{mouraAlves2004multipartite,cardy2011measuring,daley2012measuring,abanin2012measuring},  pure-state entanglement can be detected by exploiting the two-copy representation of subsystem purities~\cite{brydges2019probing, vanEnk2012measuring,elben2018renyi,elben2019statistsical,huang2020predicting},
		as is now routinely done in various experimental platforms~\cite{brydges2019probing,satzinger2021realizing,stricker2022experimental,zhu2022cross,hoke2023measurement}. RMs have been also applied to study both mixed-state bipartite entanglement, based on the estimation of the so-called partial-transposed (PT) moments~\cite{zhou2020single,elben2020mixed}, and multipartite entanglement, as characterized by the quantum Fisher information~\cite{cerezo2021sub,yu2021experimental,rath2021quantum,vitale2023estimation}.
	
	Despite these developments, estimating global properties of quantum systems with a very large number of qubits $N$ remains a difficult task. In particular, while RM approaches to estimate any \emph{local} observable are more efficient than performing full state tomography~\cite{flammia2012quantum,haah2016sample,o2016efficient, paini2019approximate, hadfield2021adaptive, huang2021efficient, hadfield2022measurements,rath2021importance,vermersch2024enhanced,yen2023deterministic, arienzo2022closed, akhtar2023scalable, bertoni2022shallow, ippoliti2023operator,ippoliti2023classical,garciaperez2021learning}, estimating \emph{global} properties requires performing exponentially-many measurements or post-processing operations. Therefore, it is unfeasible to significantly scale up existing protocols to probe global purities and bipartite entanglement~\cite{brydges2019probing, vanEnk2012measuring,elben2018renyi,elben2019statistsical,huang2020predicting,elben2020mixed}, raising the question of whether these quantities will be experimentally accessible at all as larger NISQ devices become available. 
	
	\begin{figure}[t]
		\includegraphics[scale=0.395]{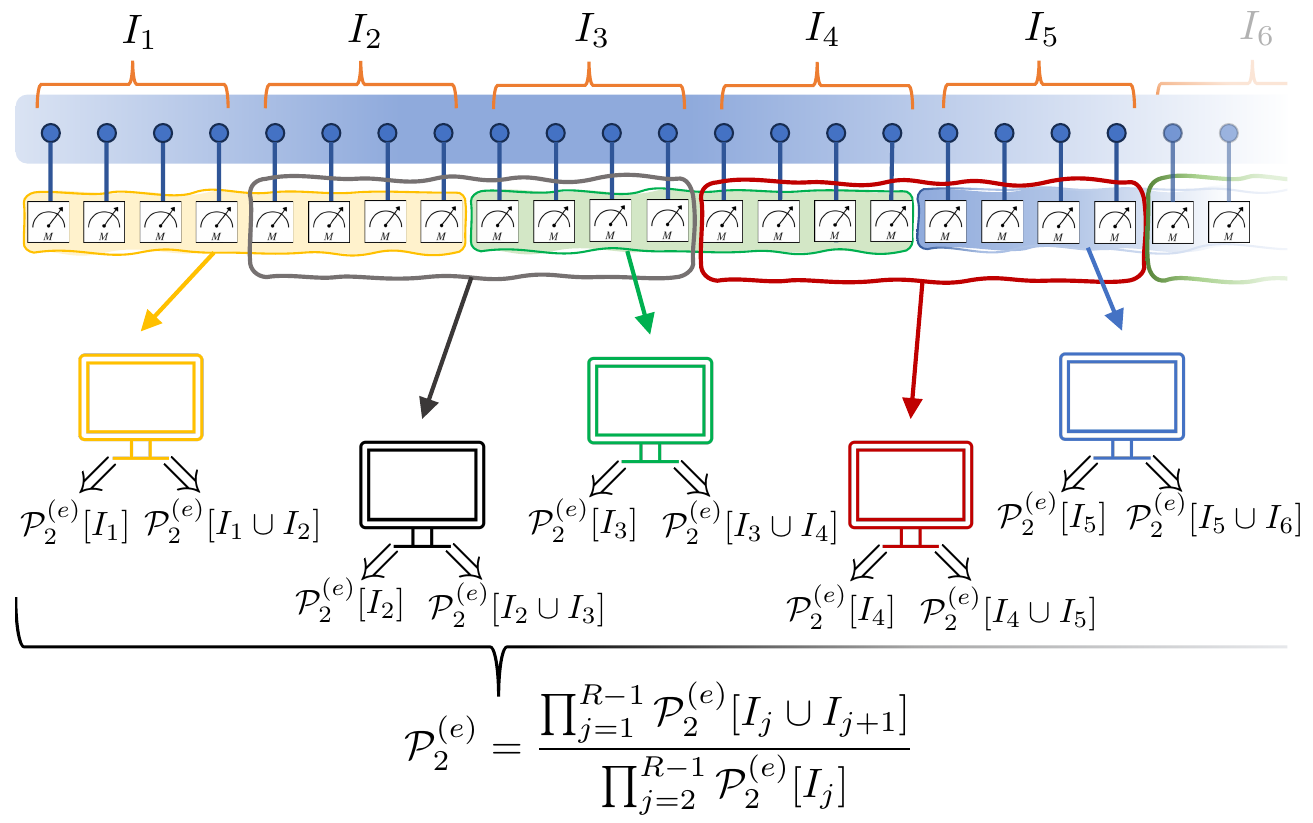}
		\caption{A $1D$ quantum system $\mathcal{S}$ is partitioned into adjacent intervals $I_j$ of size $k$, $|I_j|=k$. In this work, we consider extracting the purities over $I_j$ and $I_j\cup I_{j+1}$ using the classical-shadow approach~\cite{huang2020predicting}. Performing $M$ local measurements with respect to randomly chosen bases and classical post-processing operations, they yield faithful estimators for ${\rm Tr}[\rho_I^2]$, denoted by $\mathcal{P}_2^{(e)}[I]$. Such subsystem-purity estimators are then combined to obtain a prediction  $\mathcal{P}_2^{(e)}$ for the global purity. The accuracy of the method is controlled by $k$ and the state correlation lengths.}
		\label{fig:cartoon}
	\end{figure}
	
	In this work, we address precisely the problem of estimating global entropies and PT moments in many-qubit systems, and show that efficient strategies exist under the assumption that all spatial correlation lengths in the system are finite. This condition encompasses a large class of physically interesting cases, including ground and thermal states of local Hamiltonians, and is thus very natural when considering NISQ devices from the point of view of quantum simulation~\cite{georgescu2014,tacchino2020quantum,daley2022practical}. 
	
	In more detail, we focus on one-dimensional ($1D$) systems and put forward a strategy for the estimation of entropies and PT moments, requiring only polynomially-many measurements and post-processing operations. Our protocol is provably accurate for states satisfying a set of \emph{approximate factorization conditions} (AFCs), which express the absence of long-range correlations and which are shown to be a general feature of short-range correlated states. The basic idea, which is conveyed in Fig.~\ref{fig:cartoon}, is that the AFCs allow one to reconstruct global purities and PT moments from local information. The efficiency of our method is independent of how local information is extracted (for instance, this could be done via full tomography of certain reduced density matrices or via protocols using multiple physical copies, as discussed later). However, we will focus on an implementation making use of the standard RM toolbox~\cite{elben2023randomized}. Our motivation is two-fold. On the one hand, there is a rich literature studying the statistical errors associated with reconstructing purities and PT moments using RM measurements, allowing us to make our estimates very explicit. On the other hand, RM schemes are a very practical tool, which is now routinely employed in various experimental platforms~\cite{elben2023randomized,elben2019statistsical,cieslinski2023analysing,huang2020predicting}.
	
	The proposed approach is different from full tomographic methods relying on prior assumptions on the system state. There, a common strategy is to assume that the latter is described by a sufficiently simple ansatz wavefunction, and estimate the values of its free parameters. This logic was first put forward in the context of  matrix-product states (MPSs)~\cite{perez2006matrix,cirac2021matrixproduct, cramer2010efficient,lanyon2017efficient}, while recently extended to matrix-product density operators (MPDOs)~\cite{qin2023stable, baumgratz2013scalable,baumgratz2013scalable2,holzapfel2018petz, lidiak2022quantum}, Gibbs states~\cite{kokail2021entanglement,joshi2023exploring,anshu2021sample,anshu2023survey, rouze2021learning,onorati2023efficient}, permutation-invariant states~\cite{toth2010permutational,moroder2012permutationally,schwemmer2014experimental}, low-rank states~\cite{gross2010quantum}, stabilizers~\cite{montanaro2017learning}, tensor- and and neural-network wavefunctions~\cite{kuzmin2023learning, torlai2018neural,zhao2023provable,rieger2024sample,carrasquilla2019reconstructing,schmale2022efficient}.  While these methods may be practically very useful, they also face drawbacks. For instance, as we discuss in more detail in Sec.~\ref{sec:purity_factorization_FDQC}, an accurate estimation of, say, global purities might require reconstructing the state up to exponential precision, thus leading to unpractical overheads in $N$. On the contrary, our approach does not learn the state wavefunction, targeting purities and PT moments directly, cf.~Fig.~\ref{fig:cartoon}. 
	
	Finally, we mention that our ideas could be extended, in some cases straightforwardly, to probe other types of quantities generally requiring exponentially many measurements, including participation entropies~\cite{luitz2014participation,stephan2009shannon,stephan2009renyi,alcaraz2013,stephan2014renyi,sierant2022universal,turkeshi2023error} or stabilizer R\'enyi entropies~\cite{leone2022stabilizer}, which were recently considered in the many-body setting~\cite{leone2022stabilizer,oliviero2022magic,haug2022quantifying,haug2023stabilizer,lami2023quantum,tarabunga2023many,niroula2023phase,turkeshi2023measuring}. In addition, while we will focus on $1D$ systems, where analytic and numerical analyses are simpler, we expect that our approach could be generalized to higher spatial dimensions. Therefore, our work also opens up a number of important directions for future research. 
	
	The rest of this manuscript is organized as follows. After reviewing a few preliminary notions and tools in Sec.~\ref{sec:preliminaries}, we start by introducing the main ideas underlying our approach in Sec.~\ref{sec:toy_model}. To this end, we consider the class of so-called \emph{finite-depth quantum-circuit}~(FDQC) states, which provide an ideal toy model for short-range correlated many-body quantum states. Their minimal structure allows us to remove unnecessary technical complications from the discussion, and present the logic of our method in the simplest possible setting. We consider both purities (Sec.~\ref{sec:purity_factorization_FDQC}) and PT moments (Sec.~\ref{sec:pt_moments}), working out efficiency performance guarantees for their estimation. 
	
	The most general form of our protocol is presented in Sec.~\ref{sec:finite-range_states}. After introducing the AFCs in  Sec.~\ref{sec:AFCs}, we rigorously derive performance guarantees for the accurate estimation of the purity and PT moments (Sec.~\ref{sec:purity_PT_estimation_AFCs}) under the assumption that the state to be measured satisfies the AFCs. We then discuss the generality of the AFCs, proving that they are satisfied in translation-invariant MPDOs (Sec.~\ref{sec:MPDOs}), and presenting numerical evidence for their validity in thermal states of local Hamiltonian, cf.~Sec.~\ref{sec:AFCs_numerics}. We also present a full classical simulation of the measurement protocol (Sec.~\ref{sec:full_simulation}), studying in a concrete example the typical number of measurements required to estimate the bipartite purity of area-law pure states.
	
	Next, in Sec.~\ref{sec:efficient_entanglement_detection} we discuss a few natural examples of highly mixed states where bipartite entanglement can be detected for large system sizes by estimating just the first few PT moments. Since the latter can be very simply extracted by our approach, we argue that our results could be practically useful to probe mixed-state entanglement in experimentally available noisy quantum platforms. Finally, we report our conclusions in Sec.~\ref{sec:outlook}, while the most technical parts of our work are consigned to several appendices. 
	
	\section{Entanglement and randomized measurements}
	\label{sec:preliminaries}
	
	\subsection{Mixed-state entanglement and PPT conditions}
	\label{sec:intro_ppt_conditions}
	
	We consider a system of $N$ qubits, denoted by $\mathcal{S}$. Given a region $I\subset \mathcal{S}$, the associated Hilbert space is $\mathcal{H}_I\simeq \mathbb{C}^{\otimes 2|I|}$, where we denoted by $|I|$ the number of qubits in $I$. We will be interested in the R\'enyi entropies of the region $I$
	\begin{equation}
		S_n[I]=\frac{1}{1-n}{\log \mathcal{P}_n[I]}\,,
	\end{equation}
	where
	\begin{equation}\label{eq:def_pn}
		\mathcal{P}_n[I]={\rm Tr}[\rho_I^n]\,,
	\end{equation}
	and $\rho_I$ is the reduced density matrix on the region $I$. For $n=2$, $\mathcal{P}_2[I]$ coincides with the purity, which is a simple probe for the subsystem entropy, with $\mathcal{P}_2[I]=1$ and $\mathcal{P}_2[I]=2^{-|I|}$ for pure and maximally mixed states, respectively.
	
	Consider now a partition of $\mathcal{S}$ into two disjoint sets $\mathcal{S}=A\cup B$, yielding $\mathcal{H}_\mathcal{S}=\mathcal{H}_A\otimes \mathcal{H}_B$. If the system is in a pure state $\ket{\psi}_{AB}\in \mathcal{H}_\mathcal{S}$, its bipartite entanglement is quantified by the R\'enyi entropies $S_n(\rho_A)$~\cite{nielsen2010quantum}. Conversely, when the state of the system is mixed, its entanglement can be quantified by the logarithmic negativity~\cite{vidal2002computable,plenio2005logarithmic}
	\begin{equation}\label{eq:negativity}
		\mathcal{E}(\rho)=\log \sum_j |\lambda_j|\,,
	\end{equation}
	where the sum is over all eigenvalues $\{\lambda_j\}_j$ of the operator $\rho_{AB}^{T_A}$, and $(\cdot)^{T_A}$ denotes partial transpose with respect to subsystem $A$. The spectrum of the PT density matrix, and hence the logarithmic negativity, is completely  fixed by the PT moments
	\begin{equation}
		\label{eq:pt_moments}
		p_n={\rm Tr}\left[ \left(\rho_{AB}^{T_A}\right)^{n} \right]\,,
	\end{equation}
	for  $n=1,2,...,{\rm dim}(\mathcal{H}_A\otimes \mathcal{H}_B)$. Note that $p_1=1$, while the second PT moment coincides with the purity, $p_2=\mathcal{P}_2$~\cite{elben2020mixed}.
	
	The importance of the PT moments is two-fold. On the one hand, they can be accessed directly via RMs~\cite{zhou2020single,elben2020mixed} (or using quantum interference~\cite{carteret2005noiseless,gray2018machine}), see Sec.~\ref{sec:RM_toolbox}. On the other hand, the knowledge of the first few PT moments is enough to certify bipartite entanglement based on the non-positivity of the partial-transposed density matrix~\cite{carteret2016estimating,elben2020mixed,neven2021symmetry,zhou2020single,yu2021optimal}, or to detect different types of entanglement structures~\cite{carrasco2022entanglement}. In this work, we will consider a particular set of conditions on the PT moments to certify bipartite entanglement, which were derived in Refs.~\cite{elben2020mixed,yu2021optimal} and which we call the $p_n$-PPT conditions. Denoting by SEP the set of separable, \emph{i.e.} not entangled, states in $\mathcal{S}$, the $p_n$-PPT conditions take the form
	\begin{equation}
		\label{eq:pn_PPT}
		\rho\in {\rm SEP}\Rightarrow p_n p_{n-2}\geq p_{n-1}^2\,.
	\end{equation}
	Therefore, when the state of the system $\rho$ violates the $p_n$-PPT conditions, it is entangled, and the difference $p_{n-1}^2-p_n p_{n-2}$ is a probe for mixed-state entanglement. Note that, for $n=3$, Eq.~\eqref{eq:pn_PPT} coincides with the relation first derived in Ref.~\cite{elben2020mixed}. These conditions are in general not optimal and are a strict subset of those derived in Ref.~\cite{yu2021optimal,neven2021symmetry}. Still, their simplicity makes them particularly convenient for our purposes. 
	
	\subsection{Randomized measurements and classical shadows}
	\label{sec:RM_toolbox}
	
	The power of RM schemes lies in the fact that they need not be tailored to a specific property of the system. Rather, one performs measurements which are randomly sampled from a fixed ensemble independent of the observable of interest. Subsequently, the outcomes are processed differently depending on the quantity to be estimated~\cite{vanEnk2012measuring,elben2018renyi,knips2020multipartite,ketterer2019characterizing}. Denoting by $\rho$ the system density matrix, this approach gives us access to all observable expectation values ${\rm Tr}[\rho O]$ and, more generally, to \emph{multi-copy} objects of the form ${\rm Tr}[\rho^{\otimes n}  O]$, where the integer $n\geq 1$ is called the copy (or replica) index.
	
	In this section, we recall the basic aspects of RMs used in our work. While the logic explained in the next sections may be implemented in different ways, we will focus on a set of protocols making use the (local) classical shadows introduced in Ref.~\cite{huang2020predicting}, a prominent element in the RM toolbox~\cite{elben2023randomized}. We briefly recall the main aspects of the formalism, while we refer to Refs.~\cite{huang2020predicting,elben2023randomized} for a thorough introduction. 
	
	In what follows, we denote by $\ket{0}_j$ and $\ket{1}_j$ the basis elements of the local computational basis corresponding to qubit $j$, spanning $\mathcal{H}_j\simeq \mathbb{C}^2$. In the classical-shadow framework, one performs a set of $M$ measurements (one per experimental run, each labeled by an integer $r$), consisting of local unitary operations $\prod_j u_j^{(r)}$ followed by a projective measurement onto the computational basis 
	\begin{equation}
		|k_1,\ldots k_N\rangle=\otimes_{j=1}^N \ket{k_j}_j\,,
	\end{equation}
	with $k_j=0,1$.  The unitaries are sampled from a Haar-random ensemble, identically and independently for each qubit $j$ and experimental run $r$, see Fig.~\ref{fig:circuit}. Denoting by $\{k^{(r)}_j\}$ the set of outcomes of this two-step process, the values $\{k^{(r)}_j\}$  and the unitaries $\{u_j^{(r)}\}$ are used to define the so-called \emph{classical shadows}
	\begin{equation}
		\label{eq:global_classical_shadow}
		\rho^{(r)}_\mathcal{S}=\bigotimes_{i\in \mathcal{S}}\left[3\left(u_i^{(r)}\right)^{\dagger}\ket{k_i^{(r)}}\bra{ k_i^{(r)}} u_i^{(r)}-\openone_2\right]\,,
	\end{equation}
	which can be classically stored in an efficient way. 
	
	As mentioned, the measurement protocol does not depend on the observable of interest. Rather, one adapts the post-processing operations on the classical shadows based on the quantity to be estimated. For instance, given any observable $O$, an estimator for its expectation value is 
	\begin{equation}
		\hat{o}=\frac{1}{M}\sum_{r=1}^M{\rm Tr}[O\rho_{\mathcal{S}}^{(r)}]\,.
	\end{equation}
	It is easy to see that $\hat{o}$ is faithful, \emph{i.e.} unbiased, while different bounds for the statistical variance of this estimator may be derived depending on the locality properties of $O$~\cite{huang2020predicting,elben2023randomized}. 
	
	\begin{figure}
		\includegraphics[scale=0.40]{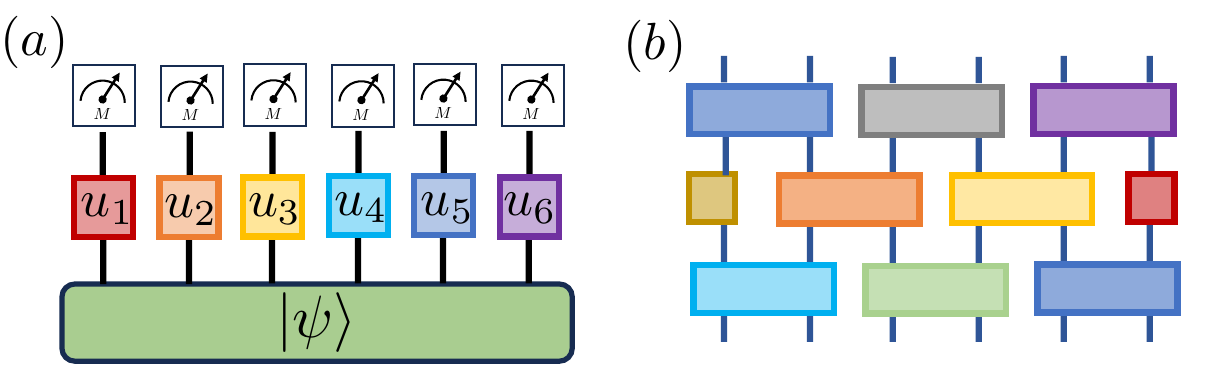}
		\caption{$(a)$: Within the classical-shadow approach, each measurement process consists of random on-site unitaries followed by local projective measurements. Their outcomes are stored and later post-processed to construct the classical shadows. $(b)$: Pictorial representation of a local finite-depth quantum circuit. The gates are arranged in a brickwork pattern, forming layers of mutually commuting unitary operations acting on pairs of neighboring qubits. Lower (upper) dangling legs correspond to the input (output) qubits. The depth of the circuit is the number of applied layers. In this picture, the depth is $\ell=3$. }
		\label{fig:circuit}
	\end{figure}
	
	It is important to recall that classical shadows also give access to entropy and PT moments. Let us consider in particular the purity, which is given by Eq.~\eqref{eq:def_pn} for $n=2$. For this quantity, the classical shadows allow us to write the estimator~\cite{huang2020predicting}
	\begin{equation}\label{eq:estimator_single_purity}
		\mathcal{P}_2^{(e)}[I]=\frac{1}{M(M-1)} \sum_{r \neq r^{\prime}} {\rm Tr}\left(\rho_{I}^{(r)} \rho_{I}^{\left(r^{\prime}\right)}\right)\,,
	\end{equation}
	where $\rho_I$ is defined as in Eq.~\eqref{eq:global_classical_shadow}. $\mathcal{P}^{(e)}[I]$ is a faithful estimator. Note that, alternatively, one can use another estimator for the purity~\cite{elben2018renyi,brydges2019probing,elben2019statistsical} that provides robustness against miscalibration errors using the same data. The statistical errors associated with $\mathcal{P}_2^{(e)}[I]$ are quantified by its variance, which can be bounded by~\cite{huang2020predicting,elben2020mixed,rath2021quantum} 
	\begin{align}\label{eq:variance_purity}
		{\rm Var}\left[\mathcal{P}_2^{(e)}[I]\right]&\leq 4\left(\frac{2^{|I|}\mathcal{P}_2[I]}{M}\right)
		+2\left(\frac{2^{2|I|}}{M-1}\right)^2\,,
	\end{align}
	where $|I|$ denotes the number of qubits in $I$. This bound is known to be essentially optimal~\cite{elben2020mixed,rath2021quantum}, telling us that an exponentially large number of measurements is needed to estimate the purity. 
	
	The PT moments~\eqref{eq:pt_moments} can be treated similarly. In this case, one constructs the estimator~\cite{elben2020mixed}
	\begin{align}\label{eq:estimator_single_PT}
		p^{(e)}_n[AB]&=\frac{1}{n !}\binom{M}{n}^{-1}\nonumber\\
		&\times \sum_{r_1 \neq r_2 \neq \ldots \neq r_n} \operatorname{Tr}\left[
		[\rho_{A B}^{\left(r_1\right)}]^{T_A} \cdots
		[\rho_{A B}^{\left(r_n\right)}]^{T_A}\right]\,.
	\end{align}
	Once again, $p^{(e)}_n[AB]$ is faithful and it is possible to derive explicit bounds on its variance, although it becomes increasingly involved for higher $n$. For instance, for $n=3$ one has~\cite{elben2020mixed,rath2021quantum}
	\begin{align}
		\label{eq:variance_PT3}
		\operatorname{Var}\left[p^{(e)}_3[AB]\right] &\leq 9 \frac{2^{|AB|}}{M} \operatorname{Tr}\left(\rho_{AB}^4\right)\nonumber\\
		+&18 \frac{2^{3 |AB|}}{(M-1)^2} p_2[AB]+6 \frac{2^{6 |AB|}}{(M-2)^3}\,.
	\end{align}
	
	Estimating the statistical errors by the variance, Eqs.~\eqref{eq:variance_purity} and~\eqref{eq:variance_PT3} imply that, in order to guarantee an accurate reconstruction of the purity and PT moments of the system, exponentially-many measurements in its size are needed. As mentioned, this makes it unfeasible to significantly scale up previous experiments making use of this strategy~\cite{brydges2019probing, vanEnk2012measuring,elben2018renyi,elben2019statistsical,huang2020predicting,elben2020mixed}. The goal of this work is to show that these limitations may be overcome under assumptions which are very common in the context of many-body physics, and thus also natural from the point of view of quantum simulation. Namely, we will put forward a set of protocols for entropy and entanglement estimation which are provably efficient assuming that all spatial correlation lengths of the state to be measured are finite. We will focus on $1D$ systems, where analytic and numerical analyses are simpler, although we expect that our approach could be generalized to higher spatial dimensions, see Sec.~\ref{sec:outlook}.
	
	\section{A toy model for the estimation protocol}
	\label{sec:toy_model}
	
	In this section we introduce the main ideas underlying our approach, focusing on a simplified setting where we can get rid of unnecessary technical complications. We analyze the case of FDQC states, where the state to be measured is prepared by a shallow (local) quantum circuit
	\begin{equation}\label{eq:FDQCs}
		\rho= U^{(\ell)}\left( \bigotimes_{j=1}^L \sigma_j \right)\left[U^{(\ell)}\right]^{\dagger}\,,
	\end{equation}
	where $L$ is length of the system, \emph{i.e.}\ the number of qubits (we use the letter $L$ instead of $N$, as in the previous section, to emphasize that we are focusing on the $1D$ case). Here $\sigma_j$ are arbitrary single-qubit density matrices, while  $U^{(\ell)}$ is a local circuit of depth $\ell$, namely $U^{(\ell)} =  V_{\ell} \cdots V_2 V_1$, where $V_{j}$ contains quantum gates acting on disjoint pairs of nearest-neighbor qubits, cf.~Fig.~\ref{fig:circuit}. We will assume that $\ell$ is fixed, \emph{i.e.}\ not increasing with the system size. We do not ask for translation symmetry and, unless specified otherwise, assume open boundary conditions. The gates making up $U^{(\ell)}$ can be arbitrary, \emph{i.e.}\ they need not be taken out of some finite gate set. The FDQC states~\eqref{eq:FDQCs} have a very simple structure from the point of view of many-body physics, but are known to approximate physically-interesting states such as MPSs~\cite{piroli2021quantum,malz2023preparation} and, more generally, ground states of gapped local Hamiltonians~\cite{chen2010local,hastings2013classifying,zeng2015quantum,zeng2015gapped} .
	
	Throughout this section,  we will assume that \emph{we know} that the state of the system is exactly of the form~\eqref{eq:FDQCs} for some finite $\ell$. We develop a protocol to efficiently estimate the purity and the PT moments of such a state, based on this knowledge. The protocol only takes the depth of the circuit, $\ell$, as an input and does not make use of state tomography. In fact, it will be later generalized replacing the assumption of the FDQC structure with the AFCs, which make no explicit reference to the form of the state wavefunction.
	
	\subsection{Purity estimation: Factorization formula}
	\label{sec:purity_factorization_FDQC}
	
	The starting point of our method is a factorization property for powers of the system density matrix. Let $I$ be any interval of adjacent qubits ($I$ can coincide with the full system $\mathcal{S}$). Consider a partition $I=A\cup B\cup C$, where $B$ separates $A$ and $C$ and denote by $|B|$ the number of qubits in $B$, cf. Fig.~\ref{fig:partition}$(a)$. Based on the fact that state is prepared by a finite depth circuit, Eq.~\eqref{eq:FDQCs}, one can prove 
	\begin{equation}\label{eq:split_1_purity}
		{\rm Tr}(\rho_{I}^2)=\frac{{\rm Tr}_{AB}(\rho_{AB}^2){\rm Tr}_{BC}(\rho_{BC}^2)}{{\rm Tr}_{B}(\rho_{B}^2)}\,,
	\end{equation}
	for any partition with $|B|\geq 2\ell-1$, where $\rho_{X}$ denotes the density matrix reduced to the interval $X$ and $XY$ is a short-hand notation for $X\cup Y$. Now, let $\{I_{j}\}_{j=1}^R$ be a collection of adjacent intervals covering the system $\mathcal{S}$, with $|I_j|=k\geq 2\ell-1$, and assume without loss of generality that $R=L/k$ is an integer, cf. Fig.~\ref{fig:cartoon}. By applying Eq.~\eqref{eq:split_1_purity} iteratively, we arrive at
	\begin{equation}\label{eq:final_product_formula_purity}
		{\rm Tr}(\rho^2)=\frac{\prod_{j=1}^{R-1}{\rm Tr}_{I_j\cup I_{j+1}}(\rho^2_{I_{j}\cup I_{j+1}})}{\prod_{j=2}^{R-1}{\rm Tr}_{I_j}\left[\rho_{I_j}^{2}\right]}\,,
	\end{equation}
	where $\rho$ is the system density matrix~\eqref{eq:FDQCs}. A proof of Eqs.~\eqref{eq:split_1_purity} and \eqref{eq:final_product_formula_purity} is given in Appendix~\ref{sec:appendix_factorization}.
	
	\begin{figure}[t]
		\includegraphics[scale=0.38]{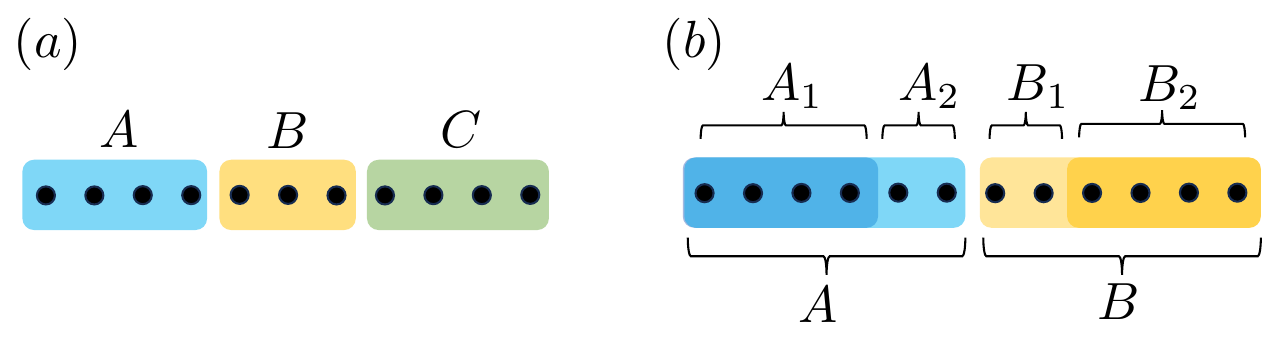}
		\caption{$(a)$: Partition considered in Eq.~\eqref{eq:split_1_purity}: A $1D$ interval $I$ is divided into three disjoint regions $A$, $B$, and $C$, where $B$ separates $A$ and $C$. $(b)$: Partition considered in Eq.~\eqref{eq:final_product_formula_PT}: A $1D$ interval $I$ is divided into two halves, $A$ and $B$. Each half is partitioned into two regions $A=A_1\cup A_2$, $B=B_1\cup B_2$}
		\label{fig:partition}
	\end{figure}
	
	Equation~\eqref{eq:final_product_formula_purity} is remarkable because its right hand side (RHS) only depends on the purities of subsystems containing up to $2k$ qubits, and thus gives us a natural basis for an efficient estimation strategy of the global purity. The idea is to first reconstruct the ``local'' purities  ${\rm Tr}[\rho_{I}^2]$ (with either $I=I_j$ or $I=I_j\cup I_{j+1}$), and subsequently plug their estimated values in the RHS of Eq.~\eqref{eq:final_product_formula_purity}, yielding the global purity estimate. This method is more efficient than directly targeting ${\rm Tr}(\rho^2)$, as the local purities can be reconstructed from a number of measurements growing exponentially in $k$, not in $L$, see Eq.~\eqref{eq:variance_purity}. Below, we give a more precise description of the protocol, proving that an accurate estimation of the purity only requires a polynomial (in $L$) number of measurements and post-processing operations.
	
	We consider performing $M_I$ measurements (\emph{i.e.}\ experimental runs) to estimate the purity of each interval $I$, with $I=I_j$ or $I=I_j\cup I_{j+1}$, see Eq.~\eqref{eq:final_product_formula_purity}. The set of measurements $M_I$ is only used to determine the purity over $I$ and for simplicity we perform the same number of measurement $M_I=M$ for each interval. Although this increases the total number of experiments $M_T$ by a factor $L$, since 
	\begin{equation}\label{eq:total_M}
		M_T=\sum_I M_I\propto L M\,,
	\end{equation}
	this procedure guarantees that statistical errors associated with distinct regions are independent and facilitate their rigorous analysis. For each interval, we then construct the estimators $\mathcal{P}^{(e)}[I]$ given in Eq.~\eqref{eq:estimator_single_purity}, and define
	\begin{equation}\label{eq:r2_e}
		r^{(e)}_2=\frac{\prod_{j=1}^{R-1}\mathcal{P}_2^{(e)}[{I_{j}\cup I_{j+1}}]}{\prod_{j=2}^{R-1}\mathcal{P}_2^{(e)}[I_j]}\,,
	\end{equation}
	which is our estimator for the global purity. Since the latter is typically exponentially small in $L$, we quantify the accuracy of $r^{(e)}_2$ by the relative error
	\begin{equation}\label{eq:purity_relative_error}
		\varepsilon_r =\left|\frac{r_2^{(e)}}{\mathcal{P}_2}-1\right|\,.
	\end{equation}
	
	Next, we seek to bound the number of measurements required  to guarantee that $\varepsilon_r$ is small with high probability. This problem is solved in Appendix~\ref{sec:relatve_error_purity}, where we prove the following result: for any arbitrarily small $\delta>0$, choosing 
	\begin{equation}\label{eq:precondition_1}
		M\geq {\rm max}\left\{2^{8k}, L^2\frac{2^{4k+10}}{k^2\delta^2}\right\}\,,
	\end{equation}
	the probability that $|\varepsilon_r|\geq \delta$ satisfies
	\begin{equation}\label{eq:main_result}
		{\rm Pr}\left[|\varepsilon_r|\geq \delta\right]\leq  \frac{2^{4k+11}L^3}{\delta^2k^3 M}\,,
	\end{equation}
	where we recall that $k=|I_j|$. The proof of Eq.~\eqref{eq:main_result} is based on a careful analysis of how the error on each factor in~\eqref{eq:r2_e} affects the global purity, making use of the statistical independence of $\mathcal{P}^{(e)}_2[I]$ for different $I$ and the so-called Chebyshev's inequality. Eqs.~\eqref{eq:total_M},~\eqref{eq:precondition_1}, and ~\eqref{eq:main_result} then imply that polynomially-many (in $L$) measurements and post-processing operations are enough to accurately estimate the purity, as anticipated.\footnote{The polynomial scaling of the post-processing operations follows straightforwardly from the definition~\eqref{eq:estimator_single_purity}, because the number of elements in each sum is $\sim M^2$, where $M=O(L^3)$.} We can rephrase this result in a more transparent way. For a given confidence level $\gamma={\rm Pr}\left[|\varepsilon_r|< \delta\right]$, Eq.~\eqref{eq:main_result} implies that it is enough to take a number of measurements
		\begin{equation}\label{eq:confidence_level_FDQC}
			M\geq \frac{2^{4k+11}L^3}{\delta^2k^3 (1-\gamma)}\,.
		\end{equation}

	Before concluding this section, a few remarks are in order. First, it is important to note that the individual factors in the formula~\eqref{eq:r2_e} could also be extracted using different RM schemes or even standard tomography for the density matrices $\rho_I$. Indeed, the possibility of estimating the global purity from polynomially-many measurements does not depend on the fact that we are using classical shadows, but rather on the factorization property~\eqref{eq:r2_e} (since its RHS only involves local subsystems). Still, classical shadows and RMs in general offer several advantages compared to state tomography. For instance, while RM schemes to estimate ${\rm Tr}[\rho_I^2]$ require a number of measurements scaling exponentially in $|I|$, the exponents are typically favorable compared to tomography~\cite{brydges2019probing,rath2023entanglement,stricker2022experimental}. In addition, classical shadows make it very easy to rigorously bound statistical errors, facilitating the analyses presented throughout this work.  
	
	Second, taking the logarithm of Eq.~\eqref{eq:final_product_formula_purity}, we obtain
	\begin{equation}\label{eq:final_formula_renyi}
		S^{(2)}(\rho)=\sum_{j=1}^{R-1}S^{(2)}(\rho_{I_j\cup I_{j+1}})-\sum_{j=2}^{R-1}S^{(2)}(\rho_{I_j})\,,
	\end{equation}
	where $	S^{(2)}(\rho)=-\log {\rm Tr}\rho^2$ is the second R\'enyi entropy.\footnote{We note that similar formulas previously appeared (albeit for the von Neumann, rather than for the R\'enyi entropy) in the context of approximate Markov-chain states~\cite{poulin2011markov,kato2019quantum,brandao2019finite, kim2017markovian}, see also~\cite{svetlichnyy2022matrix,svetlichnyy2022decay,haag2023typical,chen2020matrix,kim2021entropy,kim2021entropy2,kuwahara2020clustering,kuwahara2021improved}.} Therefore, our results can equivalently be formulated in terms of R\'enyi entropies, rather than purities (note that a small relative error on the latter implies a small additive error on the R\'enyi entropy). 
	
	Finally, as the state~\eqref{eq:FDQCs} can be represented exactly as a matrix-product operator (MPO)~\cite{silvi2019tensor}, it is instructive to compare our strategy with those based on MPDO tomography  ~\cite{qin2023stable,baumgratz2013scalable,baumgratz2013scalable2,baumgratz2013scalable2,holzapfel2018petz,lidiak2022quantum}. In many cases, these methods can efficiently provide an estimate of the system density matrix $\rho^{(e)}$, satisfying
	\begin{equation}\label{eq:norm_1}
		||\rho^{(e)}-\rho||_1\leq \delta\,,
	\end{equation}
	where $||\cdot ||_1$ denotes the trace norm, while $\delta =O(L^{-\alpha})$ is a small parameter vanishing polynomially in $L$. While this approximation allows us to accurately estimate the expectation value of any local observable, it might not be enough to extract the global R\'enyi-$2$ entropy. Indeed, given $\rho$ and $\sigma$ with $\delta=||\rho-\sigma||_1$, we have the following bound, which is known to be tight in general~\cite{audenaert2007sharp,zhihua2017sharp}
	\begin{align}\label{eq:inequality_renyi}
		|S^{(2)}(\rho)\!-\!S^{(2)}(\sigma)| &\leq 2^{L}\left[1\!-\!(1\!-\!2\delta)^2\right.
		\!-\!\left.\frac{4\delta^2}{\left(2^L\!-\!1\right)}\right]\nonumber\\
		&\sim 2^{L}\delta\,.
	\end{align}
	Therefore, a precise estimation of the R\'enyi-$2$ entropy requires an exponentially accurate reconstruction of the system density matrix, typically leading to unpractical overheads in $L$~\footnote{We note that a precise estimation of the von Neumann entanglement entropy, instead, only requires to reconstruct the target state up to an error which decays polynomially in the system size~\cite{nielsen2000continuity}. However, computing the von Neumann entanglement entropy for MPDOs is expected to be computationally hard in general.}. Instead, our method gets around this technical issue, as it does not rely on state tomography.  
	
	\subsection{The normalized PT-moment estimation}
	\label{sec:pt_moments}
	
	The ideas presented in the previous section may be applied to the PT moments, although a few subtleties must be taken into account. Denoting by $p^{(e)}_n$ our estimate for $p_n$, one is tempted to ask for a protocol that makes the relative error $|p^{(e)}_n/p_n-1|$ sufficiently small. This is, however, problematic: contrary to the moments $\mathcal{P}_n$, it is non-trivial to bound $|p_n|$ from below by a positive number. In order to get around this issue, we define the normalized PT moment
	\begin{equation}\label{eq:tilde_p_n_def}
		\tilde{p}_n[AB]=\frac{p_n[AB]}{\mathcal{P}_n[A]\mathcal{P}_n[B]}\,,
	\end{equation}
	and, instead of focusing on the relative error, we ask for an accurate estimate of $\tilde{p}_n$ up to a small \emph{additive} error. The motivation for this choice is two-fold. First, as it will be clear later, for FDQC states one can show that $\tilde{p}_n[AB]$ is independent of the system size. Therefore, contrary to the purity, one does not have to deal with numbers which are exponentially small in $L$. Second, we will see in Sec.~\ref{sec:efficient_entanglement_detection} that entanglement certification based on the $p_n$-PPT conditions  requires approximating $\tilde{p}_n[AB]$ up to a small additive error. 
	
	In order to estimate the normalized PT moment~\eqref{eq:tilde_p_n_def}, we once again rely on certain factorization properties of the density matrix. Consider the FDQC state~\eqref{eq:FDQCs} and take a partition of the system as in Fig.~\ref{fig:partition}$(b)$, with $|B_1|= |A_2|=k\geq 2\ell-1$, where $\ell$ is the depth of the circuit. As we show in Appendix~\ref{sec:factorization_pt_moments}, one can prove
	\begin{align}
		\label{eq:final_product_formula_PT}
		\!{\rm Tr}\left[\!\left(\rho_{AB}^{T_A}\right)^n\right]\!=\!{\rm Tr}\left[\!\left(\rho_{A_2 B_1}^{T_{A_2}}\right)^n\right]\! \frac{{\rm Tr}_A[\rho_{A}^n]{\rm Tr}_B[\rho_{B}^n]}{{\rm Tr}_{A_2}[\rho_{A_2}^n]{\rm Tr}_{B_1}[\rho_{B_1}^n]}.
	\end{align}
	This equation implies
	\begin{equation}
		\label{eq:factorization_sn}
		\tilde{p}_n[AB]=s_n[A_2B_1]\,,
	\end{equation}
	where we defined
	\begin{equation}\label{eq:def_sn}
		s_n[XY]=\frac{{\rm Tr}\left[\left(\rho_{XY}^{T_X}\right)^n\right]}{{\rm Tr}_{X}[\rho_{X}^n]{\rm Tr}_{Y}[\rho_{Y}^n]}\,.
	\end{equation}
	Note that $s_n[A_2B_1]$ only depends on the density matrix of a local subsystem and does not scale with the system size, as anticipated.\footnote{Because of the FDQC structure, the reduced density matrix over the region $A_2B_1$ is independent of the system size for $L\geq |A_2B_1|+2\ell-1$.} Therefore, it is possible to obtain an accurate estimate from a number of measurements and post-processing operations independent of $L$, for any arbitrary small additive error.
	
	To see this explicitly, we consider the following protocol. We introduce an estimator for $s_n[A_2B_1]$, namely
	\begin{equation}\label{eq:estimator_s}
		s^{(e)}_n=\frac{p^{(e)}_n[A_2B_1]}{\mathcal{P}^{(e)}_n[A_2]\mathcal{P}^{(e)}_n[B_1]}\,.
	\end{equation}
	The estimators for the moments $\mathcal{P}^{(e)}_n[A_2]$, $\mathcal{P}^{(e)}_n[B_1]$ and the PT-moment $p^{(e)}_n[A_2B_1]$ are defined in Eq.~\eqref{eq:estimator_single_PT}. We compute each of them out of $M$ different classical shadows, so that the total number of experimental runs is $M_T=3M$. This is done so that the statistical errors are independent, facilitating their analysis. After these steps, we obtain an estimate for $s_n[A_2B_1]$, and thus, due to Eq.~\eqref{eq:final_product_formula_PT}, for $\tilde{p}_n[AB]$. 
	
	Importantly, we can bound the additive statistical error which affects our estimate, namely
	\begin{equation}
		\varepsilon_{a} = \left|s_n^{(e)}-\tilde{p}_n[AB]\right|\,.
	\end{equation}
	For instance, focusing for simplicity on the case $n=3$, we prove in Appendix~\ref{sec:additive_error_PT} the following result: for any small $\delta>0$ and choosing
	\begin{equation}
		M\geq 27\frac{2^{11k+9}}{\delta^2}\,,
	\end{equation}
	we have
	\begin{equation}\label{eq:prob_inequality_PT}
		{\rm Pr}\left[\left|s^{(e)}_3-\tilde{p}_3[AB]\right|\geq \delta\right]\leq  81\frac{2^{11k+9}}{M\delta^2}\,,
	\end{equation}
	where $|A_2|=|B_1|=k\geq 2\ell-1$, with $\ell$ being the circuit depth. This result is proved by the same techniques used to derive Eq.~\eqref{eq:main_result}. Again, it is useful to rephrase our result in terms of the confidence level $\gamma$, yielding
		\begin{equation}\label{eq:confidence_level_PT_moment_FQDC}
			M\geq 81\frac{2^{11k+9}}{(1-\gamma)\delta^2}\,.
		\end{equation}

	Equation~\eqref{eq:prob_inequality_PT} states that a number of measurements not scaling with the system size is enough to guarantee that $\tilde{p}_3$ is approximated with arbitrary precision and high probability. Explicit performance guarantees such as~\eqref{eq:prob_inequality_PT} for higher integer values of $n$ are more cumbersome to derive, and will be omitted. Still, based on the analysis presented in Appendix~\ref{sec:additive_error_PT}, one can easily see that inequalities of the form~\eqref{eq:prob_inequality_PT} hold for higher $n$ too, where the RHS is always independent of $L$.
	
	\section{Finite-range correlated states}
	\label{sec:finite-range_states}
	
	In this section we present and discuss the most general form of our protocols. First, in Sec.~\ref{sec:AFCs} we give the definition of the AFCs, which state that the $n$-th powers of the system density matrix effectively factorize over regions of size smaller than some length scales $\xi_n$. Our definition might appear arbitrary at first, but we show later that such AFCs hold in large classes of states (see Secs.~\ref{sec:MPDOs} and ~\ref{sec:AFCs_numerics}). Next, in Sec.~\ref{sec:purity_PT_estimation_AFCs} we describe a protocol for purity and PT moment estimation in states satisfying the AFCs. Such protocols take as an input the value of the length scales $\xi_n$ (assumed to be known) and the desired accuracy, yielding as an output the numerical estimations for the purity and PT moments. The protocols only require polynomially-many (in $L$) measurements and post-processing operations. In addition, we can rigorously prove that the estimation errors are smaller than the desired threshold with high probability. Finally, we discuss the generality of the AFCs, proving analytically that they hold for MPDOs (Sec.~\ref{sec:MPDOs}) and presenting numerical evidence for their validity in Gibbs states of local Hamiltonians (Sec.~\ref{sec:AFCs_numerics}). 
	
	\subsection{The approximate factorization conditions}
	\label{sec:AFCs}
	
	The strategies developed so far rely on exact factorization conditions on the system density matrix, which are related to the sharp light-cone structure of correlation functions in FDQC states~\cite{farrelly2020review,piroli2020quantum,piroli2021quantum}. Away from these ideal cases, Eqs.~\eqref{eq:split_1_purity} and~\eqref{eq:final_product_formula_PT} do not hold, but it is natural to conjecture that they can be modified taking into account exponential corrections over scales governed by the system correlation lengths. This is done in this section, where we introduce a set of AFCs generalizing Eqs.~\eqref{eq:split_1_purity} and~\eqref{eq:final_product_formula_PT}, respectively.
	
	Given a state $\rho$, we say that it satisfies the \emph{purity AFC} if there exist  $\alpha_2, \beta_2, k_c> 0$ such that, for any interval $I=A\cup B\cup C$ as in Fig.~\ref{fig:partition}$(a)$, with $|B|=k\geq k_c$, we have
	\begin{align}\label{eq:approx_split_1}
		\!\!\!\left|{\rm Tr}(\rho_{I}^2)^{-1}\frac{{\rm Tr}_{AB}(\rho_{AB}^2){\rm Tr}_{BC}(\rho_{BC}^2)}{{\rm Tr}_{B}(\rho_{B}^2)}\!-1\right|\leq\alpha_2 e^{-\beta_2 |B|}.
	\end{align} 
	Similarly, we say that $\rho$ satisfies the \emph{PT-moment AFCs} if there exist $\alpha_n,\beta_n>0$ and an integer $k_c$ such that
	\begin{equation}\label{eq:approx_split_2}
		\left|\tilde{p}_n[AB]- s_n[A_2B_1]\right|\leq \alpha_n  e^{-\beta_n (|A_2|+|B_1|)} \,, 
	\end{equation}
	for all partitions as in Fig.~\ref{fig:partition}$(b)$ with $|B_1|= |A_2|=k\geq k_c$. Here $\tilde{p}_n[AB]$ and $s_n[A_2B_1]$ are defined in Eqs.~\eqref{eq:tilde_p_n_def} and ~\eqref{eq:def_sn}, respectively. Equations~\eqref{eq:approx_split_1} and~\eqref{eq:approx_split_2}  obviously generalize~\eqref{eq:split_1_purity} and~\eqref{eq:final_product_formula_PT}, introducing exponential corrections over length scales $\xi_n=\beta^{-1}_n$.
	
	While the definitions~\eqref{eq:approx_split_1} and~\eqref{eq:approx_split_2} might appear arbitrary at first, we show later that such AFCs hold in large classes of states (see Secs.~\ref{sec:MPDOs} and ~\ref{sec:AFCs_numerics}). This statement is very intuitive: on the one hand, we have shown that they hold exactly for FDQC states; on the other hand, in the absence of topological order ground states are known to be well-approximated by the latter~\cite{chen2010local,hastings2013classifying,zeng2015quantum,zeng2015gapped}. It is then natural to expect that the AFCs continue to hold as the temperature is increased or as the system is perturbed by incoherent noise, as these ingredients do not introduce long-range quantum correlations. In the next section, we assume that the AFCs hold, and describe a protocol for purity and PT moment estimation. We stress that we do not need any additional assumption on the state to be measured. For instance, we do not require that it can be represented efficiently by an MPDO. 
	
	Our protocols take as an input the values of $\alpha_n$, $\beta_n$, and the desired accuracy. For instance, the purity estimation protocol takes as an input $\alpha_2$, the correlation length $\xi_2=\beta^{-1}_2$, and the target threshold relative error $\delta$. Importantly, it is not necessary to know the values of $\alpha_2$ and $\beta_2$ exactly: it is sufficient to have two estimates $\tilde\alpha_2$, $\tilde\xi_2$, which upper-bound them. This is because if $\alpha_2\leq \tilde{\alpha}_2$ and $\xi_2\leq \tilde{\xi}_2$, then Eq.~\eqref{eq:approx_split_1} also holds replacing $\alpha_2$ and $\beta_2$ with $\tilde\alpha_2$ and $\tilde\beta_2=\tilde{\xi}_2^{-1}$, respectively. We can choose $\tilde{\alpha}_2$ and $\tilde{\xi}_2$ arbitrarily, but the number of measurements and post-processing operations increase parametrically with $\tilde{\alpha}_2$ and $\tilde{\xi}_2$, see Eq.~\eqref{eq:final_result_purity_short_range}. Therefore, the more accurately we can estimate the values of $\alpha_2$ and $\beta_2$, \emph{i.e.} the more information we have on the state to be measured, the more efficiently we can estimate the purity and, similarly, the PT moments. It is important to note that, in this respect, the situation is completely analogous to MPDO tomography. In that case, one \emph{assumes} that the density matrix of the system $\rho$ can be described accurately by some MPO with a bond dimension $D$. Such bond dimension does not need to be known exactly, but only an upper bound for it needs to be known. In practice, one can take $\alpha_2$ and $\beta_2^{-1}$ larger than any expected length scale in the system.
	
	Finally, suppose that, for given state $\rho$, we have an ansatz for $\alpha_2$ and $\xi_2$. From the experimental point of view, an important question is whether it is possible to efficiently certify that Eq.~\eqref{eq:approx_split_1} holds, with the ansatz values $\alpha_2$ and $\xi_2$, for the state to be measured. While at the moment this is an open question, a simple consistency check consists in repeating the purity estimation protocol (explained in \ref{sec:purity_PT_estimation_AFCs}) for increasing values of the input ansatz values $\alpha_2$ and $\xi_2$, and checking that the estimated value for the purity does not change, up to the expected precision. If the estimated purity does change, this is a ``red flag'' signalling the failure of the AFCs. We refer to Sec.~\ref{sec:outlook} for further discussions on the possibility to efficiently certify the validity of the AFCs.
	
	\subsection{Estimation protocols for states satisfying the AFCs}
	\label{sec:purity_PT_estimation_AFCs}

	First, we consider estimating the global purity of a state satisfying the AFCs~\eqref{eq:approx_split_1}. We assume the following:
	\begin{itemize}
		\item we know that the system is prepared in a state $\rho$ satisfying the AFCs~\eqref{eq:approx_split_1};
		\item we know two numerical values $\alpha_2$ and $\beta_2$ for which Eq.~\eqref{eq:approx_split_1} is satisfied.
	\end{itemize}
	Below, we detail our protocol to efficiently estimate the purity. The proof that the protocol works is non-trivial and is presented in detail in Appendix~\ref{sec:purity_short_range}. Intuitively, we use the fact that, because of the AFC~\eqref{eq:approx_split_1}, the estimator $r^{(e)}_2$ defined in Eq.~\eqref{eq:r2_e} yields a good approximation for the purity, for values of $|I_j|$ which grow mildly (logarithmically) in $L$.
	
	The protocol takes as an input the values of $\alpha_2$ and $\beta_2$, and consists of the following steps:
	\begin{enumerate}
		\item Choose the desired relative error $\delta$ on the purity. This can be any arbitrarily small number $\delta>0$;
		\item Take a partition of the system as in Fig.~\ref{fig:cartoon}, where $|I_j|=k$ and choose
		\begin{equation}\label{eq:global_approximation_condition}
			k \geq \xi_2\log(7\alpha_2 L/\delta)\,,
		\end{equation}
		where $\xi_2=\beta_2^{-1}$. The reason for the functional form appearing in the RHS is technical and explained in Appendix~\ref{sec:purity_short_range};
		\item As in Sec.~\ref{sec:purity_factorization_FDQC}, perform $M_I=M$ classical-shadow measurements to estimate the purity of each interval $I$, with $I=I_j$ or $I=I_j\cup I_{j+1}$ [we use the set of measurements $M_I$ to only determine the purity over $I$, so that the total number of experimental runs is given by~\eqref{eq:total_M}].
		\item From the classical shadows obtained out of the measurements, compute $r^{(e)}_2$ defined Eq.~\eqref{eq:r2_e}. This is the output of our protocol, giving us an estimate for the global purity.
	\end{enumerate}
	
	As mentioned, it is intuitive that the output of this protocol $r^{(e)}_2$ is, with high probability, an accurate approximation for the purity. This is proven rigorously in Appendix~\ref{sec:purity_short_range}, where we show\footnote{More precisely, Eq.~\eqref{eq:final_result_purity_short_range} holds for $M$ sufficiently large, namely $\!M\!\geq {\rm max}\left\{
		\left(7\alpha_2 L/\delta\right)^{8\xi_2\log 2}\!, (L^27^22^{10}/\delta^2)\left(7\alpha_2 L/\delta\right)^{4\xi_2\log 2}
		\right\}$, see Appendix~\ref{sec:purity_short_range}.} 
	\begin{equation}\label{eq:final_result_purity_short_range}
		\!\!{\rm Pr}\left[\left|(r_2^{(e)}/\mathcal{P}_2)-1\right|\geq \delta\right]\!\leq \frac{7^2 2^{11}L^3}{\delta^2 M}\left(\frac{7\alpha_2 L}{\delta}\right)^{4\xi_2 \log 2}\,.
	\end{equation}
	Therefore, by a polynomial number of measurements $M$ we can guarantee that the probability that the relative error $|r_2^{(e)}/\mathcal{P}_2-1|$ is larger than $\delta$ is arbitrarily small. Rephrasing our result in terms of the confidence level $\gamma$, we thus obtain
		\begin{equation}\label{eq:confidence_level_purity_AFC}
			M\geq \frac{7^2 2^{11}L^3}{\delta^2 (1-\gamma)}\left(\frac{7\alpha_2 L}{\delta}\right)^{4\xi_2 \log 2}\,.
		\end{equation}
	
	This formalizes the anticipated result. 
	
	Note that $r^{(e)}_2$ defined above is a faithful estimator for the quantity
	\begin{equation}\label{eq:def_r2}
		r_2=\frac{\prod_{j=1}^{R-1}{\rm Tr}_{I_j\cup I_{j+1}}(\rho^2_{I_{j}\cup I_{j+1}})}{\prod_{j=2}^{R-1}{\rm Tr}_{I_j}\left[\rho_{I_j}^{2}\right]}\,,
	\end{equation}
	but because~\eqref{eq:final_product_formula_purity} does not hold exactly, it is not a faithful estimator of the global purity $\mathcal{P}_2$. Still, it is a good approximation, as it is clear from the relation
	\begin{equation}\label{eq:global_approximation}
		\left|\frac{r_2}{\mathcal{P}_2}-1\right|\leq  \frac{4\alpha_2L}{k}e^{-\beta_2 k}\,,
	\end{equation}
	which holds for
	\begin{equation}\label{eq:large_k_condition}
		k\geq \xi_2\log (2\alpha_2 L)\,,
	\end{equation}
	and follows directly from~\eqref{eq:approx_split_1}, cf. Appendix~\ref{sec:purity_short_range}. In practice, the RHS of~\eqref{eq:global_approximation} makes it necessary to consider intervals $I_j$ whose length $k$ grows logarithmically in $L$ to keep the error below the desired threshold. 
	
	A similar protocol works for the PT moments. In this case, a number of measurements which does not scale with the system size is enough to guarantee arbitrary accuracy with high probability, cf. Appendix~\ref{sec:purity_short_range}. We note that the results of this section imply that the number of post-processing operations to estimate the purity and PT moments is also at most polynomial in $L$. This is because the estimators~\eqref{eq:estimator_single_purity} and~\eqref{eq:estimator_single_PT} are constructed summing a number of terms which is polynomial in $M$, and hence at most polynomial in $L$. 
	
	In the rest of this section, we discuss the generality of the AFCs. We prove analytically that they hold for the important class of translation-invariant MPDOs and provide numerical evidence that they are also satisfied by thermal states of local Hamiltonians. These results showcase the generality and versatility of our approach. 
	
	\subsection{Proof of AFCs in Matrix Product Density Operators}
	\label{sec:MPDOs}
	
	We now focus on MPDOs, a very general class of states providing accurate approximations for several short-range correlated density matrices, including thermal states of local Hamiltonians~\cite{cirac2017matrix,cirac2021matrixproduct}. We study the translation-invariant case, where we can provide analytic results. We recall that a translation-invariant MPDO $\sigma$ is defined by a four-index tensor $A^{j,k}_{a,b}$, with $j,k=0,1$ and $a, b=0,1\ldots \chi-1$, where $\chi$ is its \emph{bond dimension}~\cite{silvi2019tensor}. For a system size $L$, its matrix elements read
	\begin{equation}\label{eq:MPDO}
		\braket{i_1,\ldots ,i_L|\sigma_L |j_1,\ldots ,j_L}\!=\!A^{i_1,j_1}_{a_1,a_2}A^{i_2,j_2}_{a_2,a_3}\cdots A^{i_L,j_L}_{a_L,a_1},
	\end{equation}
	where repeated indices are summed over. This defines a family of normalized states\footnote{We assume that the local tensors $A$ generate a positive operator for all system sizes $L$, namely $\sigma_L\geq 0$ for all $L$.} 
	\begin{equation}
		\rho_L=\frac{\sigma_L}{{\rm Tr[\sigma_L]}}\,.
	\end{equation} 
	Using the standard notation from tensor-network theory~\cite{silvi2019tensor}, we can represent Eq.~\eqref{eq:MPDO} as
	\begin{equation}\label{eq:rho_tensor}
		\sigma_L = \raisebox{-11pt}{\includegraphics[scale=0.21]{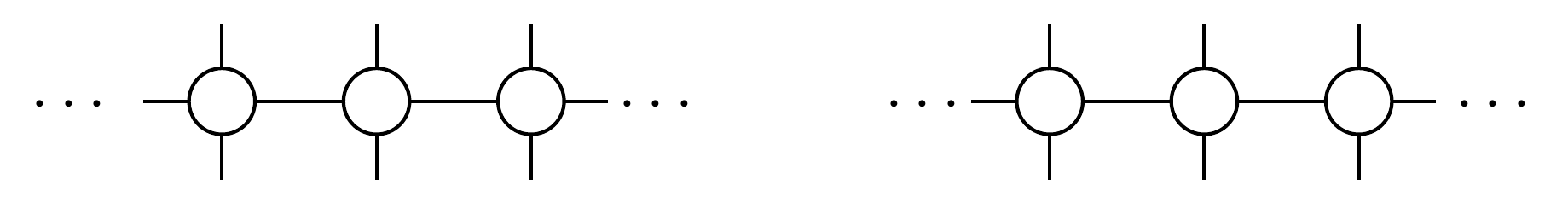}}
	\end{equation}
	Here each circle correspond to a tensor $A^{j,k}_{a,b}$. Lower and upper legs are associated with input and output degrees of freedom, respectively, while the remaining ones correspond to the ``virtual'' indices $a_j=0,\ldots \chi-1$. Note that contracted legs indicate pairs of indices which are summed over. 
	
	We first focus on the purity. We are able to show that MPDOs satisfy the purity AFCs under a few technical assumptions, which encode the fact that all correlation lengths are finite but are otherwise very general. Technically, we impose a few conditions on the spectrum and eigenstates of the transfer matrices\footnote{Note that $\tau_n$ is not Hermitian and that, for $n=1$, we recover the standard transfer matrix~\cite{cirac2017matrix}.} 
	\begin{equation}\label{eq:tau_n}
		\tau_n=
		\raisebox{-45pt}{\includegraphics[scale=0.44]{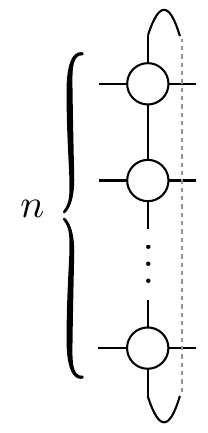}}\,,
	\end{equation}
	where the top and bottom legs are contracted. Informally, we ask that $\tau_1$ and $\tau_2$ have non-degenerate eigenvalues with largest absolute value, and that the corresponding eigenstates are not orthogonal. These conditions are general in the sense that one needs to choose fine-tuned examples to violate them. We refer the reader to Appendix~\ref{sec:factorization_MPDOs} for a precise statement and a detailed discussion. 
	
	Under these assumptions, we show that $\rho_L$ satisfies the purity AFCs. Note that, different from the rest of this work, here we are considering periodic boundary conditions. Since Eq.~\eqref{eq:approx_split_1} was introduced for open boundary conditions, we consider directly its global version~\eqref{eq:global_approximation}, which can immediately be generalized to the periodic case. Clearly, the expression for $r_2$ must be modified with respect to~\eqref{eq:def_r2}: taking into account periodic boundary conditions, we replace it with
	\begin{equation}\label{eq:r_2_periodi}
		r_2=\frac{\prod_{j=1}^{R}{\rm Tr}_{I_j\cup I_{j+1}}(\sigma^2_{I_{j}\cup I_{j+1}})}{\prod_{j=1}^{R}{\rm Tr}_{I_j}\left[\sigma_{I_j}^{2}\right]}\,,
	\end{equation}
	where we assume without loss of generality that $R=L/k$ is an integer, while the normalization factors ${\rm Tr}[\sigma^2]$ cancel out. 
	
	\begin{figure*}
		\includegraphics[width=0.95\textwidth]{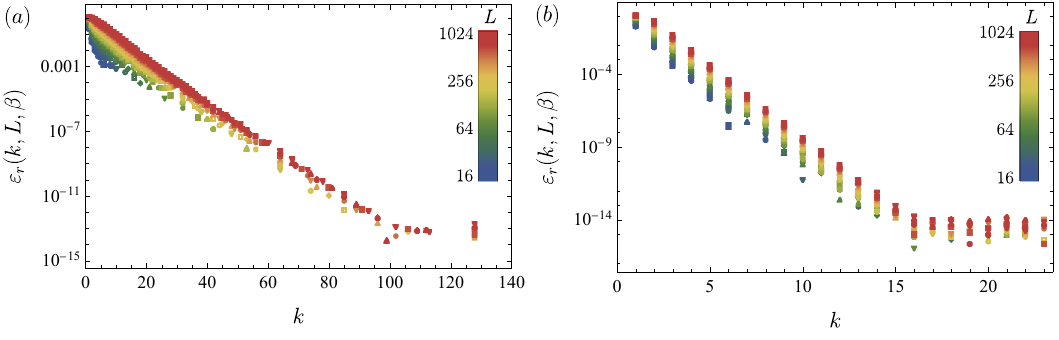}
		\caption{Relative error $\varepsilon_r(k,L,\beta)$, defined in Eq.~\eqref{eq:relative_error_numerics}, as a function of $k$, for $\beta=2$ and different values of $L$. In the left and right plots we report numerical data corresponding to the Ising and XXZ spin chains, respectively. In the former case, we chose $h_x=1.1$, $h_z=-0.04$, while the maximal bond dimension used in the computation is $\chi=32$. In the latter case, we set $\Delta=2$, while $\chi=64$.}
		\label{fig:decayingK}
	\end{figure*}
	
	The task of proving Eq.~\eqref{eq:global_approximation} [using the definition~\eqref{eq:r_2_periodi}] is carried out in Appendix~\ref{sec:factorization_MPDOs}. There, we derive the following statement: given the family of MPDOs $\rho_L$ as above, there exist $\zeta>0$, and $C>0$ (independent of $L$) such that, for any integer $k\geq k_{\rm min}$, with
	\begin{equation}\label{eq:final_kc}
		k_{\rm min}={\rm max}\left\{1, \zeta \log (20 C \chi^2 L)\right\}\,,
	\end{equation}
	we have
	\begin{equation}\label{eq:MPDO_approximation}
		\left| \frac{r_2}{\mathcal{P}_2}-1  \right| \leq \chi^2 (80C+32) \frac{L}{k} e^{-k /\zeta}\,,
	\end{equation}
	for all $L\geq {\rm max}\{\zeta\log(2^{5}\chi^2), 4k +\zeta\log (2\chi)\}$. The correlation length $\zeta$ and the constant $C$ depend on both the eigenvalues and eigenvectors of the transfer matrices $\tau_1$ and $\tau_2$, cf.~Appendix~\ref{sec:factorization_MPDOs}. Therefore, the purity approximate factorization formula~\eqref{eq:global_approximation} holds with $\alpha_2=\chi^2(20C+8)$, and $\beta_2=1/\zeta$. In turn, using the results of the previous sections, this implies that the purity of MPDOs can be estimated efficiently using our approach.\footnote{Note that, with these identifications, Eq.~\eqref{eq:global_approximation_condition} implies $k\geq k_{\rm min}$, where $k_{\rm min}$ is defined in Eq.~\eqref{eq:final_kc}.}
	
	Similarly, the AFCs for PT moments could be verified analytically for bulk-translation-invariant MPDOs with suitable open boundary conditions. This is slightly more involved, as a few additional technical assumptions on the boundaries are needed. For this reason, we do not discuss this explicitly. Instead, the AFCs for PT moments will be numerically analyzed in detail in the next section, together with the purity AFCs, for the physically interesting case of Gibbs states of local Hamiltonians in a few concrete examples.
	
	\subsection{Numerical study of AFCs in Gibbs states}
	\label{sec:AFCs_numerics}
	
	\begin{figure*}
		\includegraphics[width=0.95\textwidth]{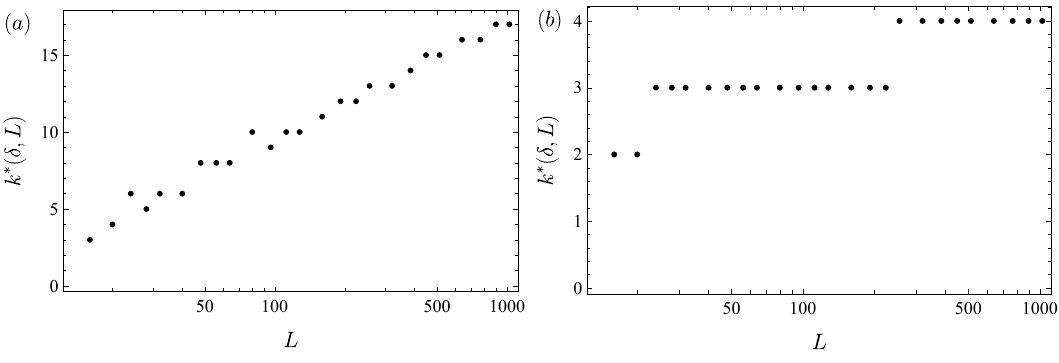}
		\caption{Numerical results for the minimal region size $k^\ast(\delta,L)$  required to achieve a specified precision $\delta$ of purity estimation, see Eq.~\eqref{kStar}. In the plots we chose $\delta=0.01$ and $\beta=2$, while the left and right figure correspond to the Ising and XXZ spin chains, respectively.  The Hamiltonian parameters and bond dimensions are the same as in Fig.~\ref{fig:decayingK}.}
		\label{fig:kStar}
	\end{figure*}
	
	As mentioned, it is natural to conjecture that the AFCs are very general, as they encode the fact that all the spatial correlation lengths in the system are finite. In this section, we provide numerical evidence supporting this claim, studying thermal states of two prototypical $1D$ local Hamiltonians. We focus in particular on the Ising model with transverse and longitudinal fields
	\begin{equation}\label{eq:ising_hamiltonian}
		H_{\rm I}=-\frac{1}{4}\sum_{j=1}^{L-1}\left[\sigma^z_j \sigma^z_{j+1}+h_x\sigma^x_j+h_z \sigma^z_j\right]\,,
	\end{equation}
	and the so-called XXZ Heisenberg spin chain
	\begin{equation}\label{eq:xxz_hamiltonian}
		H_{\rm XXZ}=-\frac{1}{4}\sum_{j=1}^{L-1}\left[\sigma^x_j \sigma^x_{j+1}+\sigma^y_j \sigma^y_{j+1}+\Delta \sigma^z_j \sigma^z_{j+1}\right]\,.
	\end{equation}
	We will study the corresponding Gibbs states
	\begin{equation}
		\rho_\beta=\frac{\exp\left(-\beta H\right)}{Z_\beta}\,,
	\end{equation}
	where $Z_\beta$ is a normalization factor.
	
	Note that the  Hamiltonians $H_I$ and $H_{\rm XXZ}$ are integrable for $h_z=0$~\cite{sachdev1999quantum} and all values of $\Delta$~\cite{korepin1997quantum}, respectively, although we do not expect that integrability plays any role in the following discussion. Note also that both Hamiltonians feature second-order quantum phase transitions at zero temperature: For $H_{\rm I}$ the critical line is $h_z=1$, while $H_{\rm XXZ}$ displays a critical phase for $|\Delta|<1$. However, we will be interested in finite-temperature states, and since in $1D$ the temperature always introduces a finite length scale, we similarly expect that quantum criticality does not play any role.\footnote{An interesting question is whether our approach could be extended to pure critical states. Although in this case there is no spatial length scale, it is known that they can be accurately approximated by MPSs with bond-dimension scaling polynomially in the system size~\cite{verstraete2006matrix}, suggesting that our approach could be generalized to include that case as well. We leave this problem for future work.} 
	
	In order to assess the validity of the AFC for the purity, we test its global version~\eqref{eq:global_approximation}. To this end, we numerically compute $r_2(k)$, as defined in Eq.~\eqref{eq:def_r2}, for increasing values of the interval size $k=|I_j|$, and the global purity $\mathcal{P}_2$. Then, we evaluate the corresponding relative error
	\begin{equation}\label{eq:relative_error_numerics}
		\varepsilon_r(k,L,\beta)=\left|\frac{r_2(k)}{\mathcal{P}_2}-1\right|\,,
	\end{equation}
	where we have made the dependence on $k$, $\beta$, and $L$ explicit. The calculations are carried out using the iTensor library~\cite{fishman2022itensor}, by first approximating the thermal states by MPOs of finite bond dimension, denoted by $\chi$, and subsequently taking powers and traces of such MPOs. For each quantity, this procedure is repeated for increasing $\chi$, until convergence is reached. We refer to Appendix~\ref{sec:appendix_numerics} for further details. 
	
	For fixed values of $L$, we first test the exponential dependence in $k$ of $\varepsilon_{r}(k,L,\beta)$ predicted by~\eqref{eq:global_approximation}. Examples of our numerical data are reported in Fig.~\ref{fig:decayingK}, showing a very clear exponential decay.  For the Ising model, we chose a non-zero value of the longitudinal field, in order to break integrability. Note that, for very large $k$, we see an apparent plateau. We attribute this behavior to finite numerical precision of our computations (note the very small values of $\varepsilon_r$ corresponding to these plateaus). It is interesting to note that the decay rate for the XXZ chain is faster than the Ising model, although higher bond dimensions are needed in order to accurately approximate the corresponding Gibbs state by an MPO. In any case, we see that quite small values of $k$ are enough to obtain a very good approximation of the purity. We have repeated the calculations for different values of the Hamiltonian parameters and $\beta$, finding consistent results.
	
	Next, we study the functional form of $\varepsilon_r(k,L,\beta)$. From the numerical point of view, it is not straightforward to test that it is asymptotically bounded by the RHS of Eq.~\eqref{eq:global_approximation}. Therefore, we proceed differently. Given an arbitrary fixed $\delta>0$, we define the minimal region size required to achieve a specified precision~$\delta$,
	\begin{equation}\label{kStar}
		k^\ast(\delta,L)={\rm min}\left\{k: \varepsilon_r(k,L,\beta)< \delta\right\}\,.
	\end{equation}
	Assuming that $\varepsilon_r$ satisfies an asymptotic bound of the form~\eqref{eq:global_approximation}, it is easy to show that $k^\ast(\delta,L)$ grows at most logarithmically in $L$, namely there exist $\alpha_2$ and $\beta_2$ such that
	\begin{equation}\label{eq:scaling_kstar}
		k^\ast(\delta,L)\leq \frac{1}{\beta_2} \log (\alpha_2 L/\delta)\,.
	\end{equation}
	When this condition holds, the estimation protocol explained in Sec.~\ref{sec:purity_PT_estimation_AFCs} is efficient. Importantly, compared to Eq.~\eqref{eq:approx_split_2},  Eq.~\eqref{eq:scaling_kstar} is straightforward to test by our numerical computations.
	
	We have verified that Eq.~\eqref{eq:scaling_kstar} is satisfied for several values of the temperature and the Hamiltonian parameters. An example of our numerical results is given in Fig.~\ref{fig:kStar}. For the Ising model, we see a very clear logarithmic scaling. For the Heisenberg chain, higher bond dimensions are required to achieve the same accuracy. Accordingly, we were not able to simulate sufficiently large system sizes to obtain a clear logarithmic behavior. Still, our results are always  consistent with a  growth of $k^\ast(\delta,L)$ which is at most logarithmic in $L$, cf. Fig.~\ref{fig:kStar}.  Thus, our numerical data fully support the validity of Eq.~\eqref{eq:scaling_kstar}. In fact, note that the values of $k^\ast(\delta, L)$ are found to be very small in practice. This is true even when the required bond dimension for MPO calculations is large, as in the case of the XXZ Heisenberg chain. This might make our approach useful already for relatively small system sizes, even if compared against MPO tomographic methods.
	
	Finally, we study the AFC for the PT moments~\eqref{eq:approx_split_1}. To this end, we define the additive error
	\begin{equation}\label{eq:additive_error}
		\varepsilon_a(k,L,\beta) =	\left|\tilde{p}_n[AB]- s_n[A_2B_1]\right|\,,
	\end{equation}
	where we have made the dependence on $k$, $\beta$, and $L$ explicit. Because the numerical cost of evaluating the trace of the $n$-th power of an MPO increases with $n$, we have restricted our analyses to $n= 3$.
	
	\begin{figure*}[t]
		\includegraphics[width=0.95\textwidth]{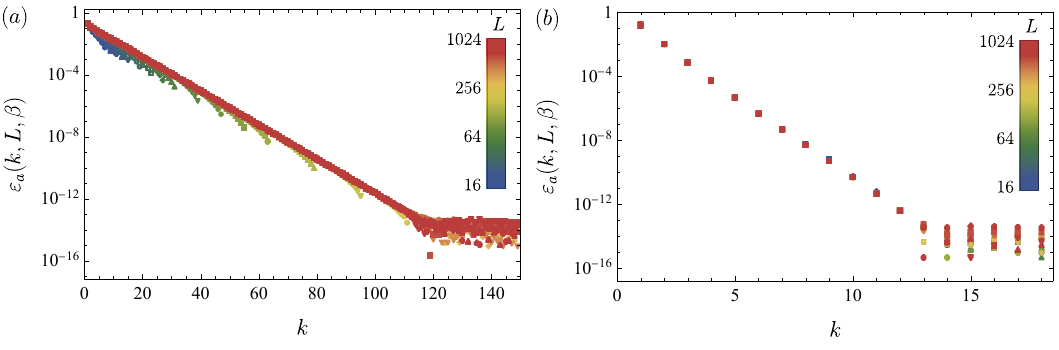}
		\caption{Additive error $\varepsilon_a(k,L,\beta)$, defined in Eq.~\eqref{eq:additive_error}, as a function of $k$, for $\beta=2$ and different values of $L$. In the left and right plots we report numerical data corresponding to the Ising and XXZ spin chains, respectively. The Hamiltonian parameters and bond dimensions are chosen as in Fig.~\ref{fig:decayingK}. }
		\label{fig:decayingK_p3}
	\end{figure*}
	
	In Fig.~\ref{fig:decayingK_p3} we report an example of our data for $\varepsilon_a(k,L,\beta)$ as a function of $k$, for increasing system sizes. From the plots, two things are apparent. First, for fixed $L$, the additive error displays a clear exponential decay. Second, its values are independent of $L$ (for values of $k$ much smaller than the system size). This behavior is perfectly consistent with Eq.~\eqref{eq:approx_split_2}. Note that, as in the previous plots, we observe irregular behavior for sufficiently large $k$, where $\varepsilon_a$ approaches a plateau and displays a $L$-dependence. We interpret these effects as arising due to numerical inaccuracies (note the very small values of $\varepsilon_a$). Similarly, when $k$ approaches $L/2$ we observe $L$-dependent behavior, which can be attributed to finite size effects. We have repeated our calculations for different $\beta$ and Hamiltonian parameters (in each case, we have checked that the bond dimension was large enough), always finding consistent results.

		\subsection{Practical feasibility of purity estimation}
		\label{sec:full_simulation}
		
		In Sec.~\ref{sec:purity_PT_estimation_AFCs}, we have shown that, for states satisfying the AFCs, the purity and the PT moments may be estimated by a number of measurements growing only polynomially in the system size. Still, the degrees of the polynomials in our bounds are high, cf. for example Eq.~\eqref{eq:confidence_level_PT_moment_FQDC}. Accordingly, one may wonder whether our approach could be applied in practice in experiments with current repetition rates. It is important to note, however, that our analysis was mainly concerned with establishing the polynomial dependence rigorously, so that our bounds may likely be improved by more refined mathematical derivations. More importantly, as we have stressed, our protocol may be improved by choosing different ways of estimating the local purities and PT moments, a simple variant being the method of ``common randomized measurements''~\cite{vermersch2024enhanced}. These improvements are expected to lead to polynomial bounds with lower exponents.
		
		Going beyond general bounds, it is important to illustrate how our estimation protocol performs in practice. This is done in this section, where we exemplify its feasibility by presenting a classical simulation of the full protocol in an explicit case. For simplicity, we focus on random MPSs~\cite{cirac2017matrix,cirac2021matrixproduct} of $N$ qubits and consider the purity of the reduced density matrix over half system. The choice of pure MPSs is motivated by the fact that it makes it easier to classically simulate the probability distribution of the measurement outcomes. At the same time, random states model typical physical states without allowing for fine-tuning. Our simulations show that the number of measurements needed to reconstruct the purities is order of magnitudes smaller compared to our theoretical bounds, demonstrating that our protocol is experimentally feasible in current quantum platforms.

		The details of our simulation are as follows. We initialize the system in a random MPS with complex coefficients and bond dimension $\chi$~\cite{fishman2022itensor}. Once the random MPS is constructed, we save a copy of it and simulate the RM process (in particular, we imagine that we have experimental access to many identical copies of the state). In our simulation, we consider a slightly more general measurement scheme compared to the one detailed in Sec.~\ref{sec:RM_toolbox}. In particular, we make use of the RMs described in Ref.~\cite{elben2018renyi} where each random unitary $u_j$ is re-used for $n_M$ local measurements. Therefore, denoting by $n_U$ the number of distinct unitaries, the total number of experiments is $n_Un_M$ (so that the classical-shadow approach described in Ref.~\ref{sec:RM_toolbox} is a special case with $n_M=1$, $n_U=M$). Each simulated experimental run produces a bit-string ${\bf s}^{(m,k)}=(s_1^{m,k},\ldots , s_L^{m,k})$, with $m=1,\ldots n_U$, $k=1\ldots n_M$. Note that, numerically, we can obtain the correct probability distribution for the measurement outcomes using the perfect-sampling algorithm for MPSs~\cite{ferris2012perfect}. We have chosen to simulate the RM scheme with re-use of local unitaries as this scheme is more efficient to simulate and it has been used in a number of recent experiments~\cite{brydges2019probing,satzinger2021realizing,stricker2022experimental,zhu2022cross,hoke2023measurement}

		Finally, we collect the simulated measurement outcomes to construct the purity estimator~\eqref{eq:r2_e}. Instead of Eq.~\eqref{eq:estimator_single_purity}, we use the following estimator for the local purities
		\begin{equation}\label{eq:Hamming_estimator}
			\mathcal{P}^{(e)}_2[I]= \frac{2^N}{n_U n_M(n_M-1)}\sum_{m=1}^{n_U} \sum_{\substack{k,k'=1\\ k\neq k'}}^{n_M} (-2)^{-D[{\bf s}_I^{(m,k)}, {\bf s}_I^{(m,k')}]}
		\end{equation}
		where $D[{\bf s}_I^{(m,k)}, {\bf s}_I^{(m,k')}]$ is the Hamming distance between the bit-stings ${\bf s}_I^{(m,k)}$ and ${\bf s}_I^{(m,k')}$ on subsystem $I$.
		Compared to~\eqref{eq:estimator_single_purity}, Eq.~\eqref{eq:Hamming_estimator} achieves a more accurate estimate for many repetitions $n_M$ using a few set of distinct unitaries $u_M$ and is expected to be more robust against miscalibration-errors~\cite{elben2023randomized}.
		We also checked numerically that our results are qualitatively unchanged when using estimators based on classical shadows.

	\begin{figure*}[t]
		\includegraphics[width=0.99\textwidth]{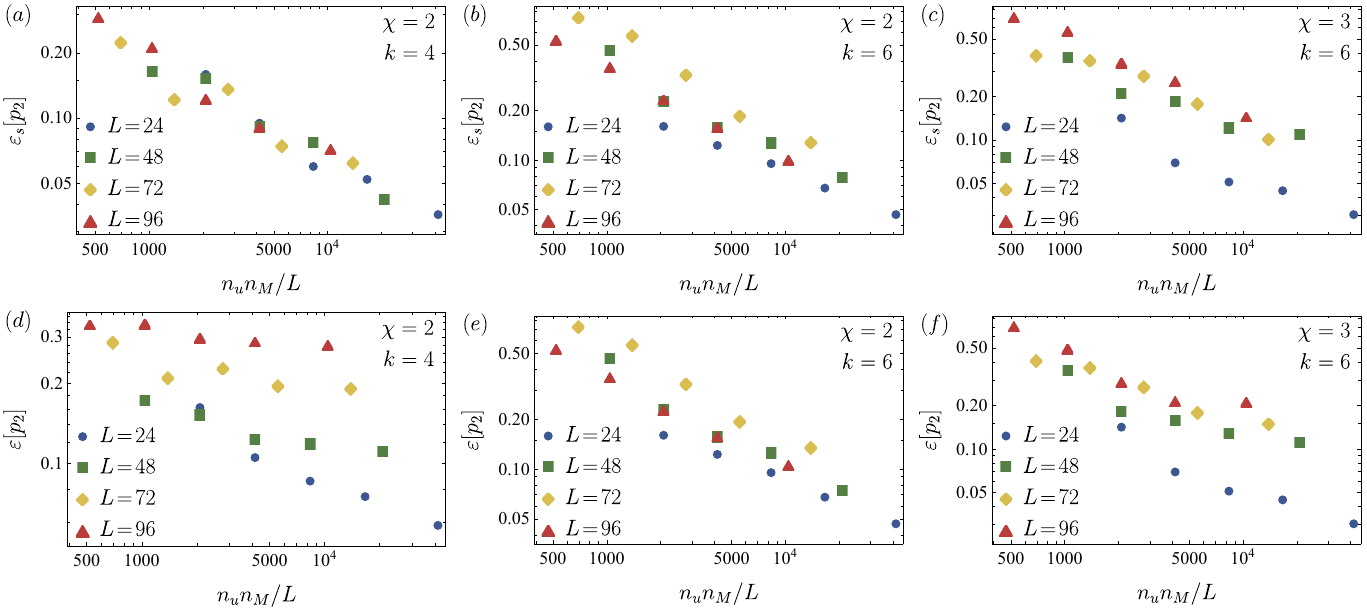}
		\caption{Statistical and full relative error on the bipartite purity in random MPSs, as a function of the number of experimental runs $n_Un_M$, with fixed $n_M=1000$. For each plot, we report the bond dimension $\chi$ of the random MPS and the length of the intervals $|I_{j}|=k$ in which the system is partitioned in our estimation protocol.}
		\label{fig:simulation_measurements}
	\end{figure*}

		Using the above simulation procedure, we have studied the estimator $r_2^{(e)}$ defined in Eq.~\eqref{eq:estimator_s} as a function of the number of measurements $n_Un_M$. We computed in particular the statistical error $\varepsilon_s= r_2^{(e)}/r_2$, where $r_2$ is the r.h.s. of Eq.~\eqref{eq:final_product_formula_purity} and the full relative error $\varepsilon=\varepsilon_r$ defined in Eq.~\eqref{eq:purity_relative_error}. In Fig.~\ref{fig:simulation_measurements}, we report our data for $\chi=2,3$, $k=4,6$ and system sizes up to $L=96$ (each plot corresponds to an average error over $20$ instances of the random MPS). 
		
		From the plots, we see that roughly $10^4$ measurements are enough to obtain a relative error below $10\%$ for system sizes up to $L\sim 100$.~\footnote{Note that in this range of the parameters, statistical and full errors are seen to be numerically close, signalling that the deviations from exact AFCs are small for $k=6$, both for $\chi=2$ and $\chi=3$.} In addition, roughly $10^5$ measurements are needed to reach an error around $5\%$ for system sizes up to $L\sim 50$ (note that, for these sizes, a direct application of RM schemes to estimate global purities is unfeasible). As announced, these numbers are order of magnitudes smaller than what predicted by our general bounds, suggesting that these bounds are far from being tight. 
		
		We stress that the number of measurements performed in our simulations are of the order of those performed in current trapped-ion and superconducting qubit experiments, see e.g. Refs.~\cite{brydges2019probing,satzinger2021realizing,stricker2022experimental,zhu2022cross,hoke2023measurement,vitale2023estimation}. The results reported in this section thus substantiate the claim that our protocol is practically useful for global purity and entanglement estimation in today's quantum platforms.

	\section{Mixed-state entanglement detection}
	\label{sec:efficient_entanglement_detection}
	
	In this section we study the $p_n$-PPT conditions for states satisfying the AFCs. In Sec.~\ref{sec:pn_PPT_fn} we first introduce a set of quantities, which we call $f_n(\rho)$, displaying the following properties: (i) $f_n(\rho)<0$ if and only if $\rho$ violates the $p_n$-PPT conditions and $(ii)$ $f_n(\rho)$ can be estimated efficiently if $\rho$ satisfies the AFCs, representing natural probes for mixed-state entanglement. Next, in Sec.~\ref{sec:examples}, we compute $f_n$ in a few concrete cases including Gibbs states of local Hamiltonians, and identify the values of the system parameters for which $f_2, f_3<0$. The results of this section show that our approach can be successful in detecting entanglement even for large numbers of qubits and highly mixed states, thus being practically useful in many situations routinely encountered in available quantum platforms.
	
	\subsection{The $p_n$-PPT conditions and mixed-state entanglement probes}
	\label{sec:pn_PPT_fn}
	
	Given the density matrix $\rho_{AB}$ on a bipartite system, its logarithmic negativity~\eqref{eq:negativity} can be extracted from the knowledge of all PT moments $p_n$, with $n=1,\ldots,{\rm dim}(\mathcal{H}_A\otimes \mathcal{H}_B)$. Therefore, while our method is efficient to estimate the individual PT moments, reconstructing the logarithmic negativity is still hard, as it requires an exponential number of them.
	
	Luckily, the $p_n$-PPT conditions introduced in Sec.~\ref{sec:intro_ppt_conditions} allow us to detect mixed-state entanglement from only a few PT moments. For example, the first two non-trivial elements in the family~\eqref{eq:pn_PPT} only involve the PT moments up to $n=5$. Explicitly, they read
	\begin{subequations}\label{eq:ppt_conditions_statement}
		\begin{align}
			p_3{\rm -PPT}: &\quad	p_3p_1\geq p_2^2\,,\label{eq:p3}\\
			p_5{\rm -PPT}: &\quad	p_3p_5\geq p_4^2\label{eq:p5}\,.
		\end{align}
	\end{subequations}
	The conditions~\eqref{eq:ppt_conditions_statement} can be equivalently rewritten in terms of the quantities
	\begin{subequations}\label{eq:ppt_conditions_functions}
		\begin{align}
			f_3(\rho_{AB})&=\tilde{p}_3\tilde{p}_1-\tilde{p}_2^2\frac{\mathcal{P}_2^{2}[A]\mathcal{P}_2^{2}[B]}{\mathcal{P}_1[A]\mathcal{P}_1[B]\mathcal{P}_3[A]\mathcal{P}_3[B]}\,,\\
			f_5(\rho_{AB})&=\tilde{p}_5\tilde{p}_3-\tilde{p}_4^2\frac{\mathcal{P}_4^{2}[A]\mathcal{P}_4^{2}[B]}{\mathcal{P}_3[A]\mathcal{P}_3[B]\mathcal{P}_5[A]\mathcal{P}_5[B]}\,,
		\end{align}
	\end{subequations}
	where $\tilde{p}_n$ are the normalized PT moments defined in Eq.~\eqref{eq:tilde_p_n_def}. Indeed $f_3$ and $f_5$ are negative if and only if the inequalities~\eqref{eq:ppt_conditions_statement} are violated.
	
	Crucially, $f_3(\rho_{AB})$ and $f_5(\rho_{AB})$ can be estimated efficiently if $\rho_{AB}$ satisfies the AFCs. Indeed, they only depend on $\mathcal{P}_n$ and the normalized PT moments $\tilde{p}_n$, both of which can be estimated using the protocol introduced in Sec.~\ref{sec:finite-range_states}.\footnote{Note that, while in Sec.~\ref{sec:purity_PT_estimation_AFCs} we only considered the purity $\mathcal{P}_2$, similar AFCs can be defined for $\mathcal{P}_n$ with $n>2$, and hence similar estimation protocols work to estimate higher moments. Note also that, for the case of FDQCs, the proof of exact factorization conditions for $\mathcal{P}_2$ can be trivially extended to any $\mathcal{P}_n$ with $n>2$, cf. Appendix~\ref{sec:appendix_factorization}.} We note that such estimation protocols allow us to reconstruct $f_n(\rho_{AB})$ up to any arbitrary small additive error. Namely, for any arbitrary positive constant $\delta$, one can construct an estimator $f^{(e)}_n$, with $|f^{(e)}_n-f_n|\leq \delta$ by polynomially many (in $L$) measurements and post-processing operations. This is easily seen recalling that $\tilde{p}_n$ ($\mathcal{P}_n$) can be estimated up to any additive (relative) error, and using that
	\begin{align}\label{eq:intermediate_inequalities}
		|\tilde{p}_n[I]|\leq 1\,,\qquad \frac{(\mathcal{P}_n[I])^{2}}{\mathcal{P}_{n-1}[I]\mathcal{P}_{n+1}[I]}\leq 1\,,
	\end{align}
	where the first inequality from $|\tilde{p}_n[I]|\leq \tilde p_2[I]=\mathcal{P}_2[I]$, while the second follows from H\"older's inequality.  
	
	At this point, it is important to ask whether the quantities $f_n$ are practically useful in the many-body context considered in this work. For example, one could be worried that the conditions~\eqref{eq:ppt_conditions_statement} are never violated for large $L$ (even if the logarithmic negativity is non-zero), or that $f_n$ is exponentially small in $L$ for highly mixed states, thus requiring an exponential accuracy for its estimation. To address this question, consider a family of states $\rho^{(L)}_{AB}$, defined for increasing system sizes $L$ and satisfying the AFCs. Suppose that we extract $f_n(\rho^{(L)}_{AB})$ using our estimation protocol for a fixed approximation error $\delta_n$. Then, the $p_n$-PPT conditions~\eqref{eq:ppt_conditions_statement} allow us to successfully detect entanglement if
	\begin{subequations}\label{eq:negativity_conditions}
		\begin{align}
			f_{3}(\rho^{(L)}_{AB})&\leq -C_3\,,\label{eq:negativity_conditions_1}\\
			f_{5}(\rho^{(L)}_{AB})&\leq -C_5\,,\label{eq:negativity_conditions_2}
		\end{align}
	\end{subequations}
	where $C_3, C_5>0$ are constants (independent of $L$) with $C_n\geq 2\delta_n$. Indeed, in these cases the statistical inaccuracy is smaller (with high probability) than the amount by which $f_n(\rho^{(L)}_{AB})$ is negative, allowing us to certify entanglement. 
	
	In the next subsection, we study a few natural examples of states satisfying the AFCs showing that~\eqref{eq:negativity_conditions} hold either for all values of $L$ or up to very large system sizes. Therefore, we exhibit concrete examples where $f_n$ are provably useful for entanglement detection.
	
	\subsection{Mixed-state entanglement detection in FDQC and Gibbs states}
	\label{sec:examples}
	
	We start by studying the relations~\eqref{eq:negativity_conditions} in the simplest case of FDQC states~\eqref{eq:FDQCs}. We consider a family of states where the single-qubit density matrices in Eq.~\eqref{eq:FDQCs} are parametrized as
	\begin{equation}
		\sigma_i(\gamma)=\gamma\ket{a_i}\bra{a_i}+(1-\gamma) \ket{b_i}\bra{b_i}\,.
	\end{equation}
	Here $\ket{a_i}$, $\ket{b_i}$ are arbitrary orthonormal states, while $0\leq \gamma\leq 1/2$ plays the role of a depolarizing parameter, controlling both purity and bipartite entanglement [$\gamma$ is not to be confused with the symbol previously used for confidence levels]. We ask for which values of $\gamma$ the conditions~\eqref{eq:ppt_conditions_statement} are violated. 
	
	This problem is analyzed in detail in Appendix~\ref{sec:ppt_condition_for_short_range}. We prove that, for generic choices of the unitary gates, there exist $\gamma_3, \gamma_5>0$ such that Eqs.~\eqref{eq:negativity_conditions_1} and ~\eqref{eq:negativity_conditions_2}  are satisfied, respectively, for all $\gamma\leq \gamma_3/L$ and $\gamma\leq \gamma_5/L^{1/3}$. Now, choosing in particular $\gamma_L= \gamma_5/L^{1/3}$, we have
	\begin{align}
		{\rm Tr}[\rho^{(L)}(\gamma_L)^2]&=[\gamma_L^2+(1-\gamma_L)^2]^L\nonumber\\
		& \leq \exp\left[-2(\gamma_L -\gamma_L^2)L\right]\nonumber\\
		&\leq  \exp\left[-\gamma_L L\right]=\exp\left[-\gamma_5 L^{2/3}\right]\,,
	\end{align}
	namely
	\begin{equation}
		S_2[\rho^{(L)}(\gamma_L)]\geq \gamma_5 L^{2/3}\,.
	\end{equation}
	Therefore, $\rho^{(L)}(\gamma_L)$ provides an example where the $p_5$-PPT condition is violated by a finite amount $C_5$ and, at the same time, the global R\'enyi-$2$ entropy grows polynomially, albeit sub-linearly, in the system size.\footnote{It is interesting to note that the $p_3$-PPT condition only allows us to detect bipartite entanglement in states with $O(1)$ entropy. Therefore, the $p_5$-PPT condition improves this result by allowing the entropy to scale as $O(L^{2/3})$.  It is then natural to conjecture that higher-$n$ PPT conditions lead to a further improvement of the maximal scaling to $O(L^{\alpha})$, with $\alpha$ approaching $1$ in the large-$n$ limit.}
	
	\begin{figure*}
		\includegraphics[width=0.95\textwidth]{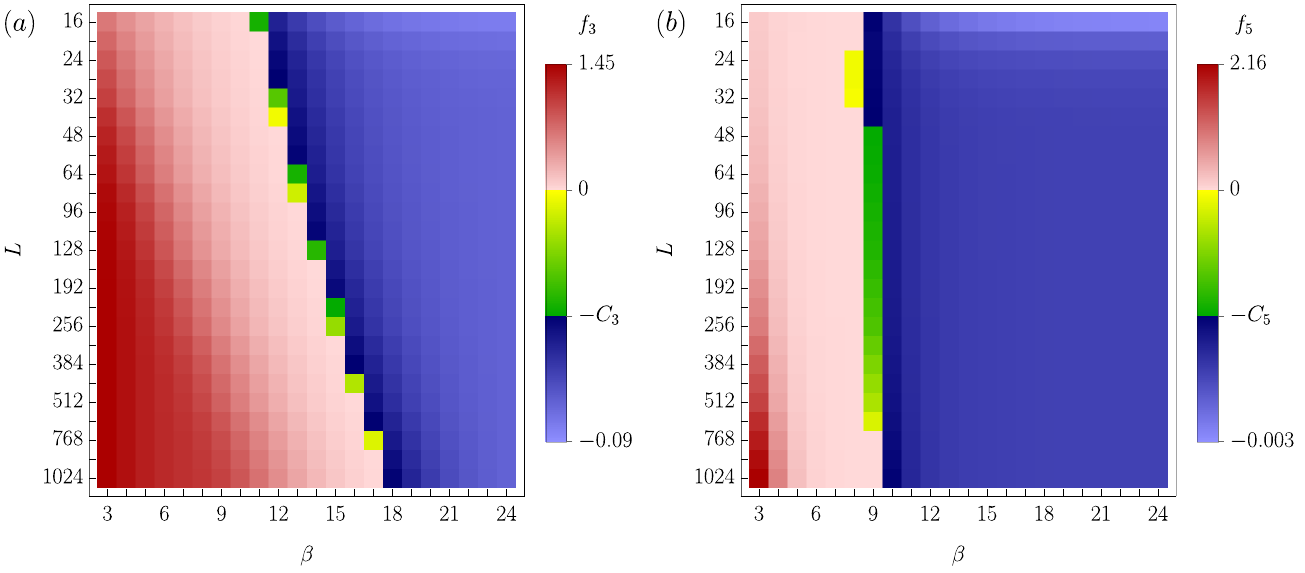}
		\caption{Density plots for the coefficients $f_3(\beta,L)$, $f_5(\beta,L)$, defined in Eq.~\eqref{eq:ppt_conditions_functions} as a function of $\beta$ and $L$. The data correspond to thermal states in the Ising model~\eqref{eq:ising_hamiltonian}, with $h_x=1.1$ and $h_z=-0.04$. Red (blue) regions indicate the values of $\beta$ and $L$ for which $f_3(\beta,L)$, $f_5(\beta,L)$ are larger (smaller) than some small values $-C_3=0.01$ and $-C_5=0.001$, respectively. }
		\label{fig:ppt_conditions}
	\end{figure*}
	
	As a final interesting example, let us consider once again finite-temperature states of local Hamiltonians. In this case, the purity decays exponentially in $L$ and we expect that the conditions~\eqref{eq:ppt_conditions_functions} fail to detect the presence of entanglement in the thermodynamic limit. In the following we give evidence that they can nevertheless be useful when considering systems of large but finite size.
	
	To support this claim, we study thermal states in the Ising model~\eqref{eq:ising_hamiltonian}, and numerically compute the coefficients~\eqref{eq:negativity_conditions} for different temperatures and system sizes. The computations are performed using tensor-network methods, following the same strategy explained in Sec.~\ref{sec:AFCs_numerics}. An example of our results is reported in Fig.~\ref{fig:ppt_conditions}. In these plots, we have fixed small but otherwise arbitrary positive numbers $C_3$ and $C_5$, and identified the values of $\beta$ and $L$ for which $	f_{3}(\beta, L)\leq -C_3$ and 	$f_{5}(\beta, L)\leq -C_5$.
	
	The data corresponding to the $p_3$-PPT condition are reported in the left plot, which clearly shows that the values of the temperature $T=\beta^{-1}$ for which the entanglement can be detected decrease with $L$, consistent with our previous analysis on FDQC states. On the other hand, the right plot shows that $f_5(\beta,L )<-C_5$ up to values of the temperature ($T\simeq 0.1)$ that do not significantly depend on $L$, at least up to very large system sizes ($L\simeq 1024$). We expect that this is a finite-size effect, namely that increasing $L$ further, the values of $T$ for which entanglement can be detected will decrease. Still, this example shows how the $p_n$-PPT conditions may be able to detect entanglement even for highly mixed states and up to very large system sizes, making the estimation protocol for PT moments practically very useful. 
	
	\section{Outlook}
	\label{sec:outlook}
	
	We have studied the problem of estimating global entropies and entanglement in many-qubit states, assuming the validity of certain AFCs which encode the fact that all the spatial correlations lengths in the system are finite. We have shown that the AFCs allow one to reconstruct entropies and PT moments from local information which can be extracted efficiently using standard RM schemes. Exploiting this fact,  we have devised a very simple strategy for entropy and entanglement estimation which is provably accurate, requiring polynomially-many local measurements and post-processing operations. We have discussed the generality of the AFCs, providing both analytic and numerical evidence for their validity in different classes of states, including MPDOs and thermal states of local Hamiltonians. 
	
	As we have argued, the protocols proposed in this work could be practically useful to detect bipartite mixed-state entanglement for large numbers of qubits in today’s quantum platforms. This is especially true when considering current NISQ devices from the point of view of quantum simulation, since the AFCs hold under assumptions that are very common in the context of many-body physics. In addition, our work also raises important questions and opens up a number of interesting directions. 
	
	First, our protocols assume that the state to be measured satisfies the AFCs. Therefore, an important question from the experimental perspective is whether it is possible to devise a complementary (and still efficient) strategy to \emph{certify} their validity, similar to what can be done in the case of MPS~\cite{cramer2010efficient} or entanglement-Hamiltonian tomography~\cite{joshi2023exploring}.
	
	As we have discussed in Sec.~\ref{sec:AFCs}, a simple but effective consistency check can be carried out by repeating our protocols multiple times, each time with an increased ansatz value for the correlation lengths. If the estimated purities or PT moments change beyond the expected precision over different repetitions, we obtain a ``red flag'' signalling the failure of the AFCs in the state to be measured. This process, however, does not allow one to rigorously rule out the presence of long-range correlations, as this would require running the estimation protocols using, as an input, correlation lengths which are proportional to the system size. We envision that a way to tackle the certification problem could be to identify a hierarchy of properties which are strictly stonger than the AFCs and for which rigorous certification protocols can be found more easily. At the moment, however, this remains an open problem.
	
	Next, while we have focused on $1D$ systems, a very natural direction is to extend our approach to higher spatial dimensions. In this case, additional technical complications arise, but we expect that efficient estimation strategies for the purity and PT moments exists under similar conditions encoding the absence of long-range correlations. We stress that this problem is particularly important, since tomographic methods for many-qubit states (based, for instance on tensor or neural networks) are much less developed in higher dimensions. 
	
	Finally, we mention that our approach could also be applied to probe other quantities which generally require exponentially many measurements. Straightforward examples include the participation entropies~\cite{luitz2014participation,stephan2009shannon,stephan2009renyi,alcaraz2013,stephan2014renyi} and the so-called stabilizer R\'enyi entropies~\cite{leone2022stabilizer}, but we expect that similar ideas could be implemented in other contexts. One example is that of cross-platform verification protocols~\cite{elben2020cross}, which involve probing the (possibly exponentially small) overlaps between different density matrices. Similarly, efficient estimation strategies for the so-called Loschmidt echo~\cite{quan2006decay} also appear to be within the reach of the techniques developed in this work. While these problems require dealing with additional technical subtleties compared to the analysis of the entropy and PT moments, we believe that the latter are not substantial and could be overcome. We leave the study of these applications to future research. 
	
	\section*{Acknowledgments}
	
	BV acknowledges funding from the Austrian Science Foundation (FWF, P 32597 N), and from the French National Research Agency via the JCJC project QRand (ANR-20-CE47-0005), and via the research programs Plan France 2030 EPIQ (ANR-22-PETQ-0007),  QUBITAF (ANR-22-PETQ-0004) and HQI (ANR-22-PNCQ-0002).
	ML and MS acknowledge support by the European Research Council (ERC) under the European Union’s Horizon 2020 research and innovation program (Grant Agreement No. 850899). MS acknowledges hospitality of KITP supported in part by the National Science Foundation under Grants No.~NSF PHY-1748958 and PHY-2309135. JIC is supported by the Hightech Agenda Bayern Plus through the Munich Quantum Valley and the German Federal Ministry of Education and Research (BMBF) through  EQUAHUMO (Grant No. 13N16066). PZ acknowledges funding from the European Union’s Horizon 2020
	research and innovation programme under grant agreement
	No 101113690 (PASQuanS2.1). \\
	
	\noindent \emph{Note Added}: After completion of this work, a work by Gondolf \emph{et. al} appeared on the arXiv~\cite{gondolf2024conditional}, showing that any translation-invariant, finite-range Hamiltonian in 1$D$ at any inverse
		temperature $\beta>0$ satisfies the purity AFCs, in complete agreement with our numerical results of Sec.~\ref{sec:AFCs_numerics} and further substantiating the generality of our approach.
	
	\appendix
	
	\section{Factorization condition for FDQC states}
	\label{sec:appendix_factorization}
	
	\subsection{The purity}
	\label{sec:proof_toy_model}
	In this Appendix we show that the FDQC states~\eqref{eq:FDQCs} satisfy the factorization condition~\eqref{eq:split_1_purity}. In fact, we prove the stronger condition
	\begin{align}\label{eq:strong_factorization}
		&{\rm Tr}_{B^{\otimes n}}\left[\overrightarrow{\Pi}_B \left(\rho_{ABC}^{\otimes n}\right)\right]\qquad&\nonumber\\
		&= \frac{{\rm Tr}_{B^{\otimes n}}\left[\overrightarrow{\Pi}_B \left(\rho_{AB}^{\otimes n}\right)\right]\otimes  {\rm Tr}_{B^{\otimes n}}\left[\overrightarrow{\Pi}_B \left(\rho_{BC}^{\otimes n}\right)\right]}{{\rm Tr}_B(\rho_B^n)}\,,
	\end{align}
	where we introduced the $n$-copy cyclic permutation operator $\overrightarrow{\Pi}_B\ket{i_1,\ldots i_n}=\ket{i_n,i_1,\ldots i_{n-1}}$, and we assumed $|B|\geq (2\ell-1)$, where $\ell$ is the circuit depth. We make use of a graphical derivation. We show the case $n=2$, $\ell=4$, and $|B|=8$ for concreteness, but it is clear that the proof generalizes for all values of $n$ and $\ell$. The reduced density matrix over the interval $\mathcal{I}=A\cup B\cup C$ reads 
	\begin{equation}\label{eq:picture_dm}
		\raisebox{-48pt}{\includegraphics[scale=0.31]{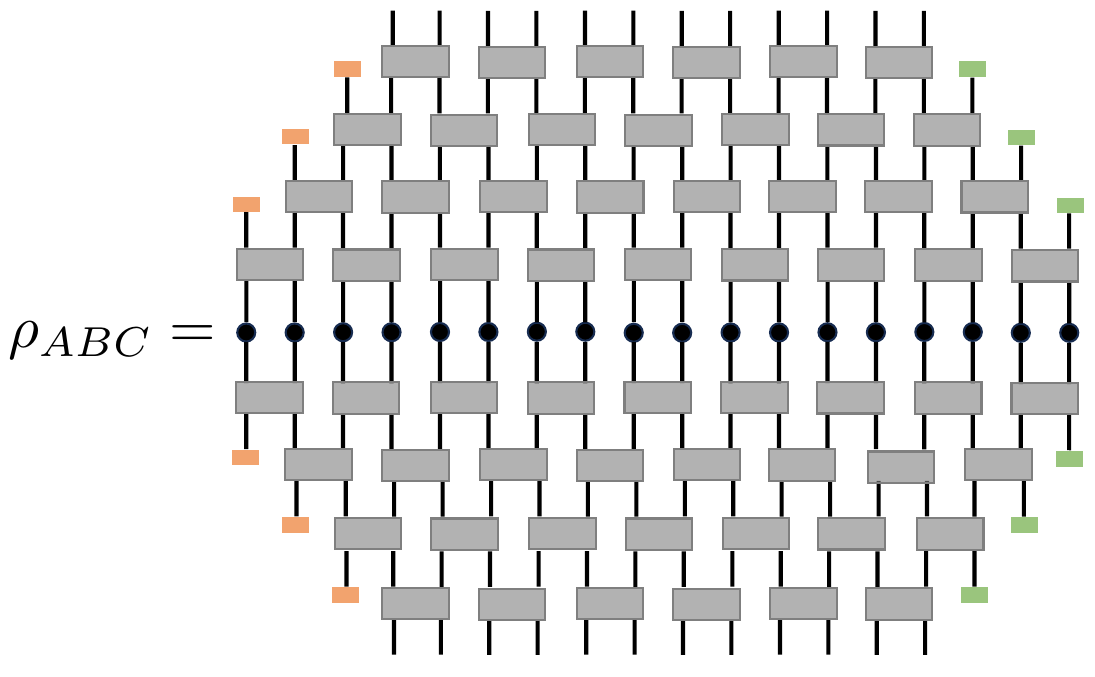}}
	\end{equation}
	Following the standard tensor-network notation, lower (upper) dangling legs correspond to the input (output) qubits. In addition, here and in the following upper and lower legs in the same column which are marked with a small rectangle of the same color are traced over. Finally, black dots correspond to the density matrices $\sigma_j$. 
	
	The left-hand side of Eq.~\eqref{eq:strong_factorization} can be represented as
	\begin{align}
		\raisebox{-35pt}{\includegraphics[scale=0.32]{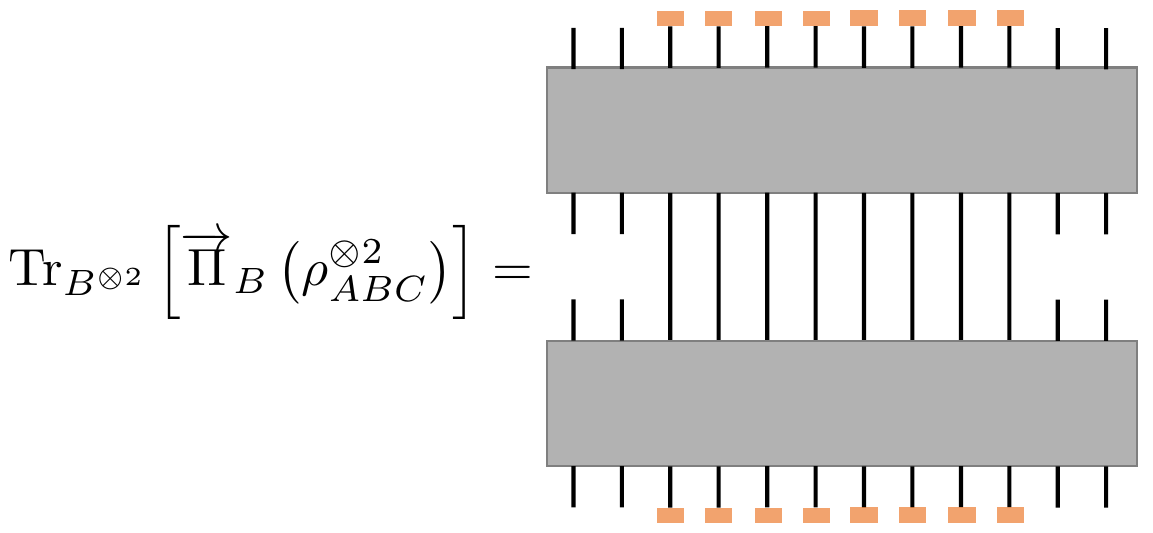}}
	\end{align}
	where dangling legs are not contracted, while the gray rectangle is a short-hand notation for the reduced density matrix~\eqref{eq:picture_dm}. Using unitarity of the gates, we get  ${\rm Tr}_{B^{\otimes n}}\left[\overrightarrow{\Pi}_B \left(\rho_{ABC}^{\otimes n}\right)\right]=T_1T_2T_3$, where
	\begin{equation*}
		\raisebox{-5pt}{\includegraphics[scale=0.24]{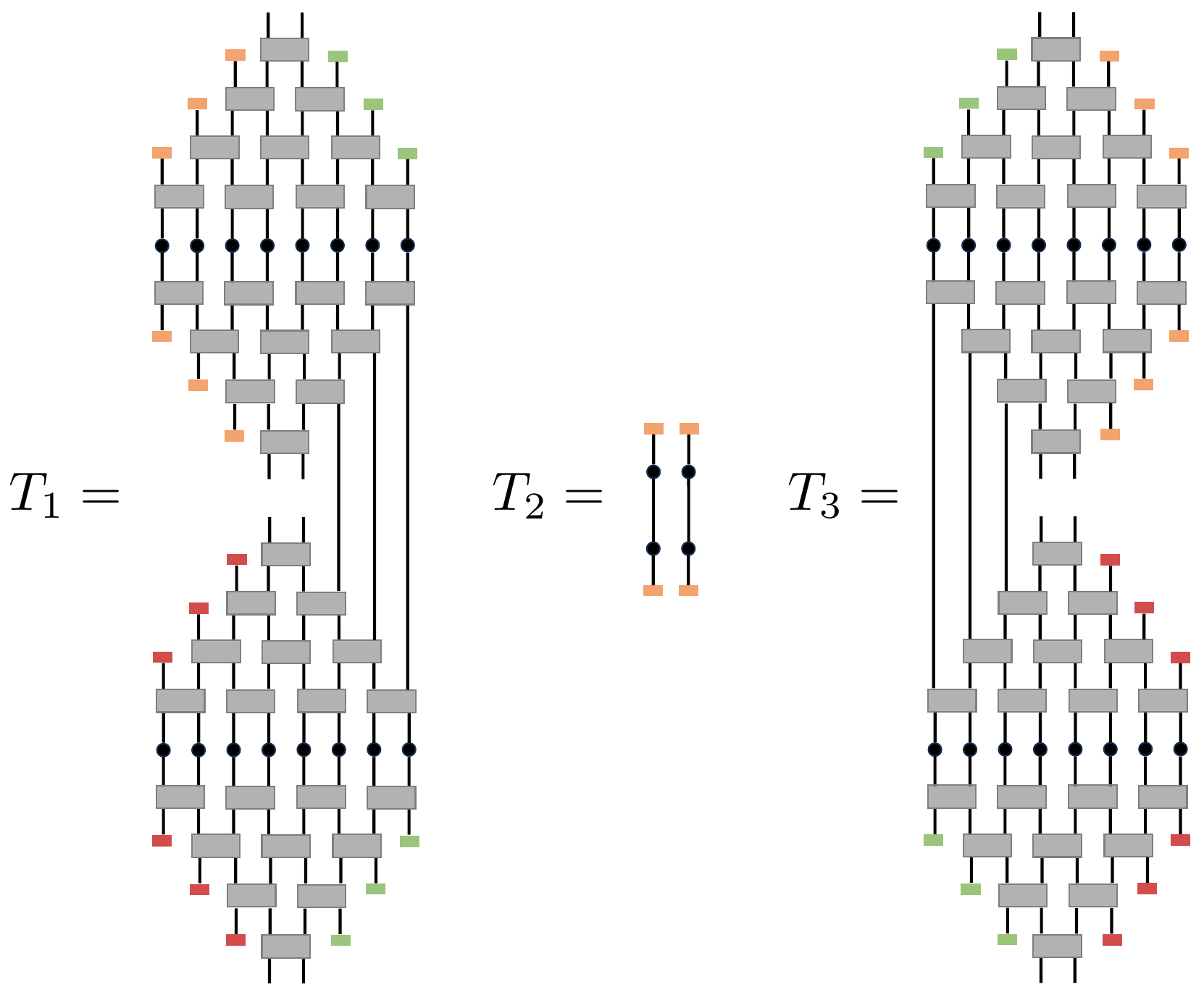}}
	\end{equation*}
	Note that $T_1$ and $T_3$ are operators, while $T_2$ is a number. Let us compare this graphical expression with the one in the RHS of Eq.~\eqref{eq:strong_factorization}. The first and second terms in the numerator read, respectively
	\begin{align*}
		\raisebox{-68pt}{\includegraphics[scale=0.24]{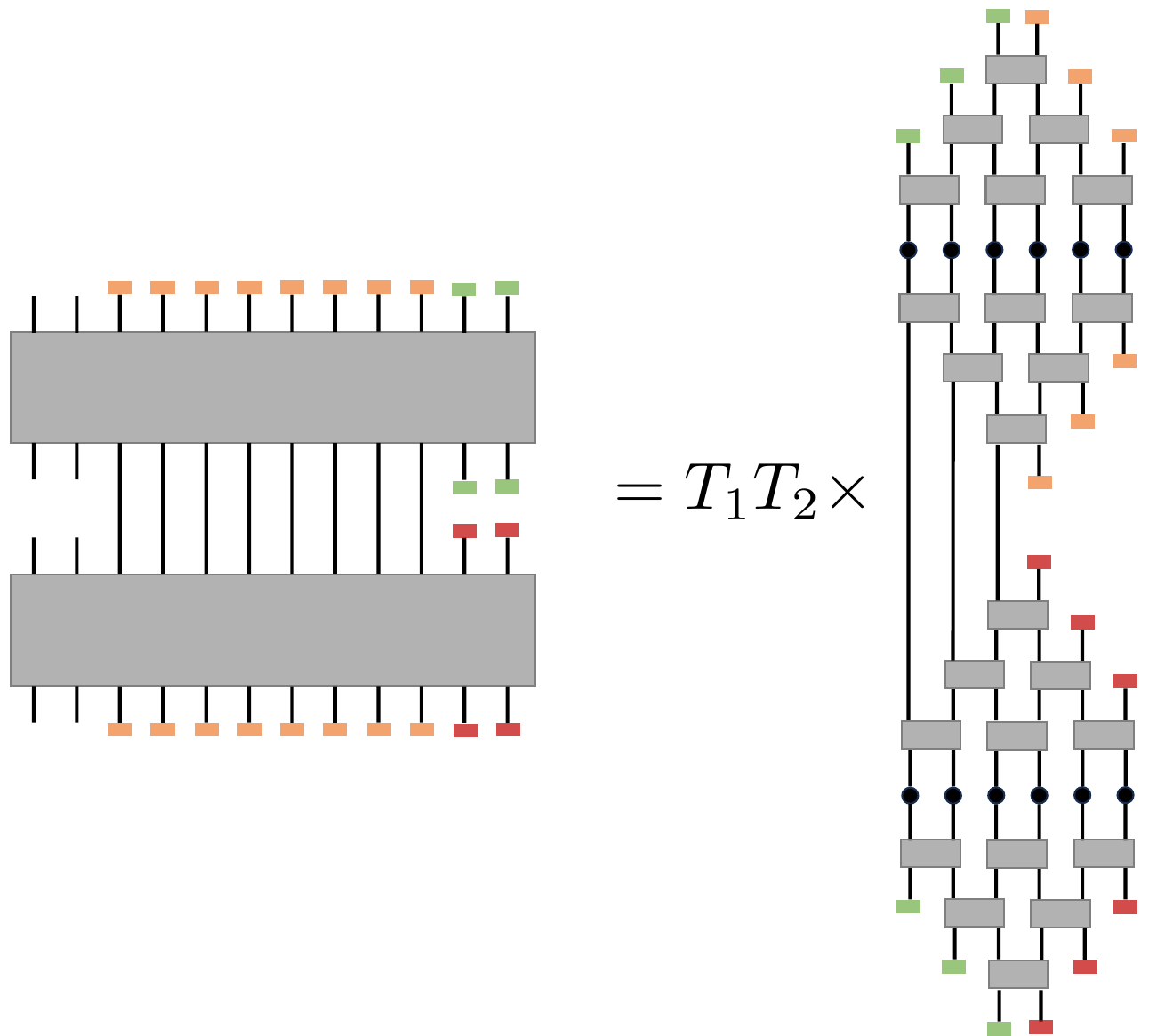}}\\
		\raisebox{-68pt}{\includegraphics[scale=0.24]{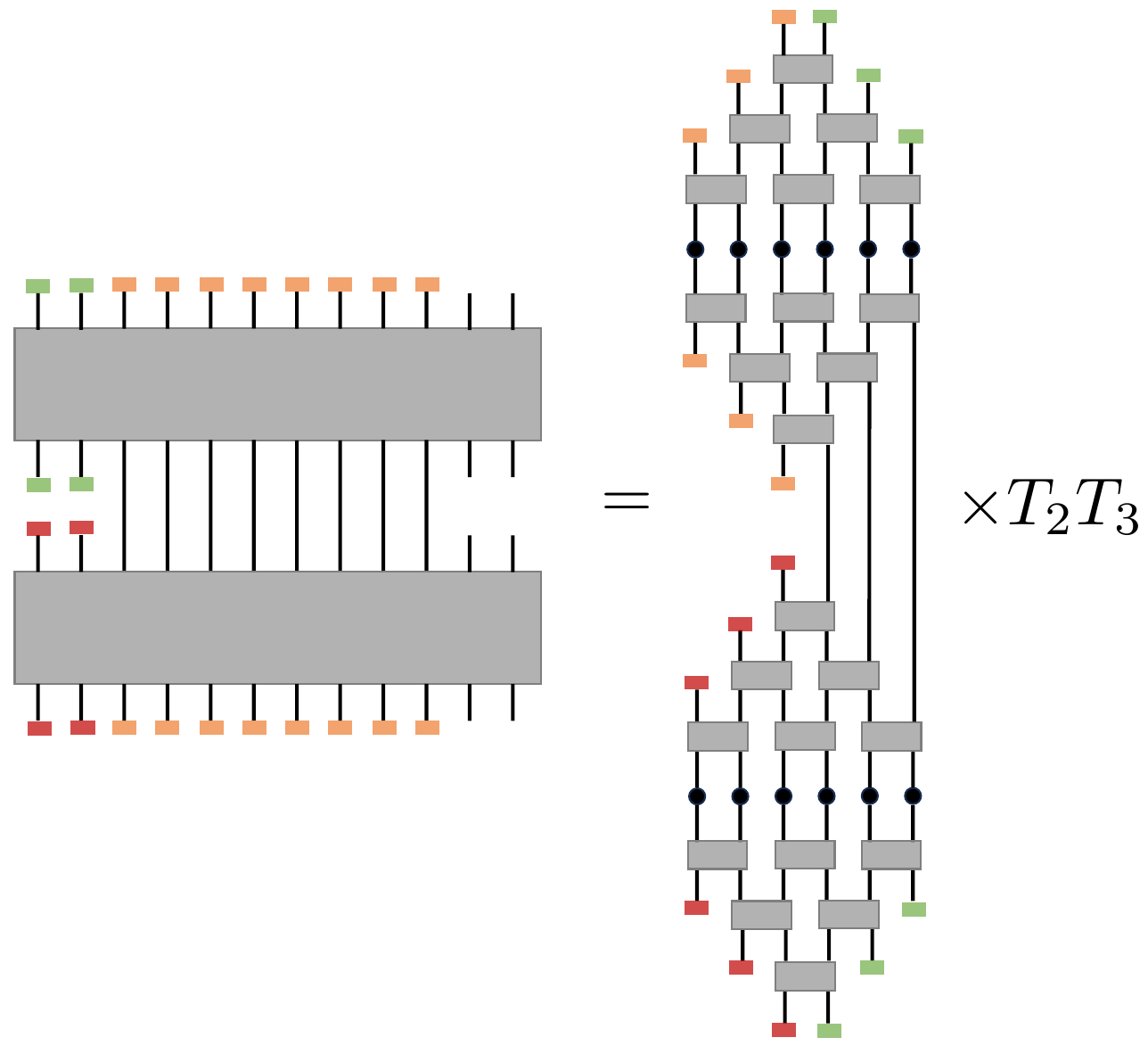}}\,.
	\end{align*}
	Similarly, the denominator in Eq.~\eqref{eq:strong_factorization} reads
	\begin{equation*}		\raisebox{-49pt}{\includegraphics[scale=0.24]{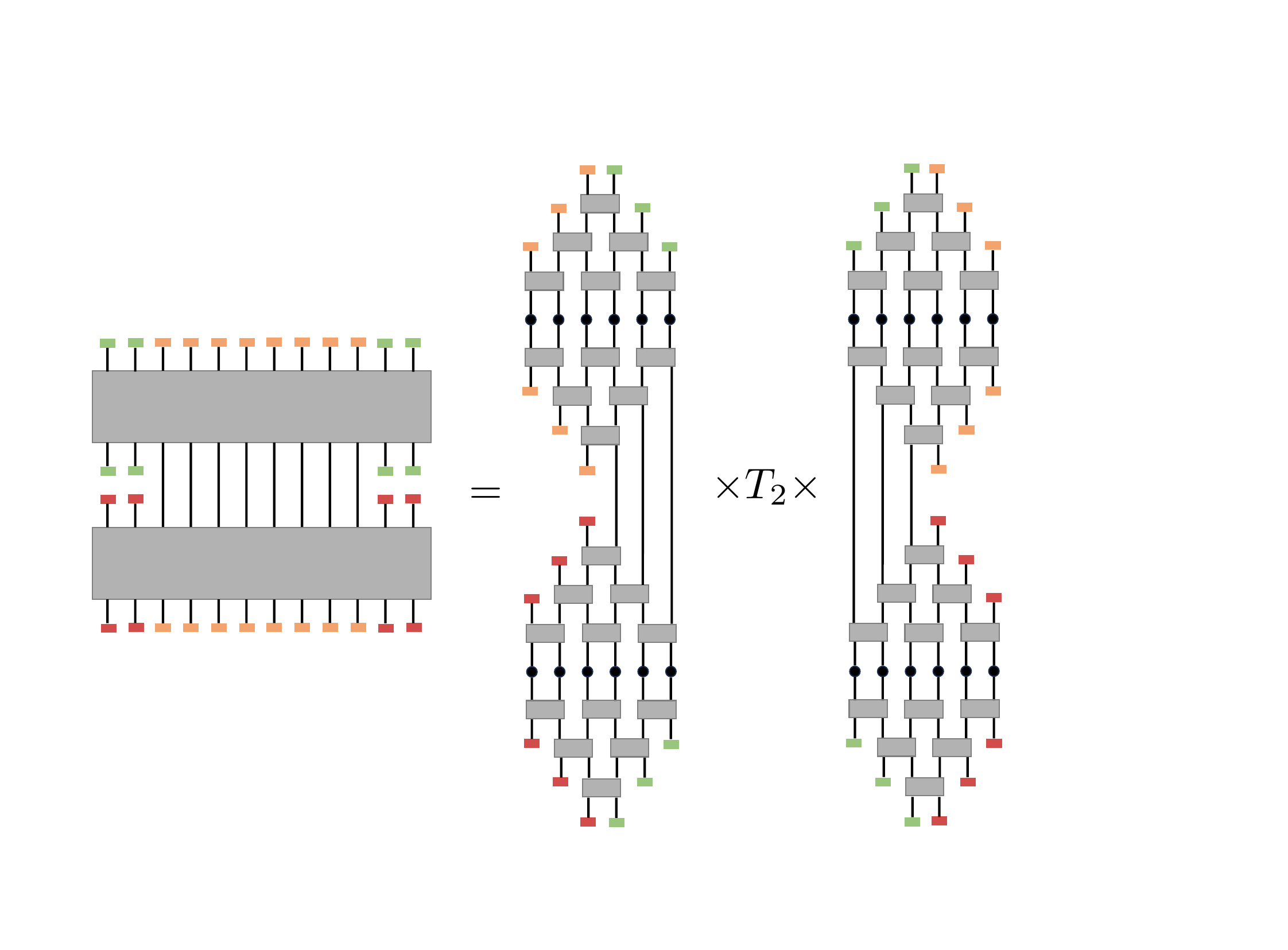}}
	\end{equation*}
	Putting everything together, and simplifying the common factors in the numerator and denominator, we arrive at Eq.~\eqref{eq:strong_factorization}.
	Finally, using the fact that 
	$\Tr(\rho_X^n)=\Tr(\overrightarrow{\Pi}_X\rho_X^{\otimes n})$, and the relation
	$\overrightarrow{\Pi}_{ABC}=\overrightarrow{\Pi}_{A}\otimes \overrightarrow{\Pi}_{B} \otimes \overrightarrow{\Pi}_{C}$, we arrive at
	\begin{equation}
		{\rm Tr}(\rho_{ABC}^n)=\frac{{\rm Tr}_{AB}(\rho_{AB}^n){\rm Tr}_{BC}(\rho_{BC}^n)}{{\rm Tr}_{B}(\rho_{B}^n)}\,.
	\end{equation}
	For $n=2$, this proves Eq.~\eqref{eq:split_1_purity}. 
	
	Next, we prove Eq.~\eqref{eq:final_product_formula_purity}. Let us consider a partition of the whole system $\mathcal{S}=A\cup B\cup C$, where $|B|=|C|=k\geq 2\ell-1$. We first make use of Eq.~\eqref{eq:split_1_purity}. Setting now $A^{(0)}:=A$, $B^{(0)}:=B$, $C^{(0)}:=C$ we can iterate this operation, each time splitting the interval $A^{(j)}\cup B^{(j)}$ into three adjacent regions $A^{(j+1)}$, $B^{(j+1)}$, and $C^{(j+1)}$ with $|B^{(j+1)}|=|C^{(j+1)}|=k$. We continue this procedure until $|A^{R}|=k$, which gives us Eq.~\eqref{eq:final_product_formula_purity}.\\
	
	\subsection{The PT moments}
	\label{sec:factorization_pt_moments}

	The goal of this Appendix is to prove Eq.~\eqref{eq:final_product_formula_PT}. 
	
	We start by recalling the expression for the PT moments~\cite{carteret2005noiseless}
	\begin{align}
		{\rm Tr}\left[\left(\rho_{AB}^{T_A}\right)^n\right]&
		=
		{\rm Tr}
		\left[
		\overleftarrow{\Pi}_A
		\overrightarrow{\Pi}_B
		\left(\rho_{AB}^{\otimes n}
		\right)
		\right],
	\end{align}
	where we introduced the $n$-copy forward (backward) cyclic permutation operator in the $n$-replica space, defined by
	\begin{align}
		\overrightarrow{\Pi}_A\ket{i_1,\ldots i_n}&=\ket{i_n,i_1,\ldots i_{n-1}}\,,\nonumber\\
		\overleftarrow{\Pi}_B\ket{i_1,\ldots i_n}&=\ket{i_2,i_3,\ldots i_{1}}\,.
	\end{align}
	Using this notation, it is easy to see that Eq.~\eqref{eq:strong_factorization} implies
	\begin{align}\label{eq:transposed_strong_factorization}
		&{\rm Tr}_{B^{\otimes n}}\left[\overleftarrow{\Pi}_B \left(\rho_{ABC}^{\otimes n}\right)\right]=\nonumber\\
		&\frac{{\rm Tr}_{B^{\otimes n}}\left[\overleftarrow{\Pi}_B \left(\rho_{AB}^{\otimes n}\right)\right]\otimes  {\rm Tr}_{B^{\otimes n}}\left[\overleftarrow{\Pi}_B \left(\rho_{BC}^{\otimes n}\right)\right]}{{\rm Tr}_B(\rho_B^n)}\,,
	\end{align}
	Eq.~\eqref{eq:transposed_strong_factorization} can be proved by taking the partial-transpose with respect to $A^{\otimes n}$ and $C^{\otimes n}$ in both sides of~\eqref{eq:strong_factorization}, followed by complex conjugation. Next, let us take a bipartition of the system into $A$ and $B$, each divided into two intervals $A=A_1\cup A_2$ and $B=B_1\cup B_2$ with $|A_2|=|B_1|=k$, cf. Fig.~\ref{fig:partition}$(b)$. Making use of Eq.~\eqref{eq:strong_factorization}, we have
	
	\begin{align}
		{\rm Tr}\left[\left(\rho_{AB}^{T_A}\right)^n\right]&={\rm Tr}_{A^{\otimes n}}{\rm Tr}_{B_2^{\otimes n}}\overleftarrow{\Pi}_A\overrightarrow{\Pi}_{B_2}{\rm Tr}_{B_1^{\otimes n}}\overrightarrow{\Pi}_{B_1}\rho^{\otimes n}\nonumber\\
		&=\frac{{\rm Tr}_{A^{\otimes n}\otimes B_1^{\otimes n}}(\overleftarrow{\Pi}_A \overrightarrow{\Pi}_{B_1}\rho_{AB_1}^{\otimes n})}{{\rm Tr}_{B_1}(\rho_{B_1}^n)}\nonumber\\
		&\times \frac{{\rm Tr}_{B_1^{\otimes n} \otimes B_2^{\otimes n}}(\vec{\Pi}_{B_1}\vec{\Pi}_{B_2}\rho_{B_1B_2}^{\otimes n})}{{\rm Tr}_{B_1}(\rho_{B_1}^n)}\,.
	\end{align}
	Using 
	\begin{align}
		&{\rm Tr}_{A^{\otimes n}\otimes B_1^{\otimes n}}(\overleftarrow{\Pi}_A \overrightarrow{\Pi}_{B_1}\rho_{AB_1}^{\otimes n})\nonumber\\
		&={\rm Tr}_{A_1^{\otimes n}}{\rm Tr}_{B_1^{\otimes n}}\overleftarrow{\Pi}_{A_1}\overrightarrow{\Pi}_{B_1}{\rm Tr}_{A_2^{\otimes n}}\overleftarrow{\Pi}_{A_2}\rho_{A_1A_2B_1}^{\otimes n}\nonumber\\
		&={\rm Tr}_{A_1^{\otimes n}\otimes A_2^{\otimes n}}(\overleftarrow{\Pi}_{A_1} \overleftarrow{\Pi}_{A_2}\rho_{A_1A_2}^{\otimes n})\nonumber\\
		&{\rm Tr}_{A_2^{\otimes n} \otimes B_1^{\otimes n}}(\overleftarrow{\Pi}_{A_2}\overrightarrow{\Pi}_{B_1}\rho_{A_2B_1}^{\otimes n})
		\left[{\rm Tr}_{A_2}(\rho_{A_2}^n)\right]^{-1}\,,
	\end{align}
	we finally arrive at Eq.~\eqref{eq:final_product_formula_PT}.
	
	\section{Statistical-error analysis in FDQC states}
	
	\subsection{Statistical-error analysis for the purity }
	\label{sec:relatve_error_purity}
	In this Appendix we prove Eq.~\eqref{eq:main_result}. We consider the protocol explained in Sec.~\ref{sec:purity_factorization_FDQC} and denote by $\mathcal{P}^{(e)}_2[I]$ the estimates for $\mathcal{P}_2[I]={\rm Tr}(\rho_I^2)$ obtained from the classical-shadow approach. These local purities are estimated with a non-zero relative error, \emph{i.e.} $\mathcal{P}_2^{(e)}[I]/\mathcal{P}_2[I]=(1+\varepsilon^{I})$. We define 
	\begin{equation}\label{eq:def_varepsilon_appendix}
		\varepsilon={\rm max}\left\{|\varepsilon^{I}|: I\in \{I_{j}\}_j\cup\{I_{j}\cup I_{j+1}\}_{j}\right\}\,.
	\end{equation}
	We will take $r^{(e)}_2$, defined in Eq.~\eqref{eq:r2_e}, as the experimental estimate for the global purity $\mathcal{P}_2={\rm Tr}(\rho^2)$, choosing the intervals $I_j$ with $|I_j|=k\geq 2\ell-1$, where $\ell$ is the circuit depth, so that Eq.~\eqref{eq:final_product_formula_purity} holds.
	
	We start our proof from a few preliminary lemmas.
	\begin{lem}\label{lem:inequality_eps}
		Suppose $\varepsilon\leq 1/2$. Then
		\begin{equation}\label{eq:inequality_eps_appendix}
			\left |\frac{r_2^{(e)}}{\mathcal{P}_2}-1\right |\leq e^{(4L/k) \varepsilon}-1\,.
		\end{equation}
	\end{lem}
	\begin{proof}
		Setting $\mathcal{P}_2^{(e)}[I]=\mathcal{P}_2[I](1+\varepsilon^{I}$),recalling $R=L/k$, where $k=|I_j|$ and using~\eqref{eq:final_product_formula_purity}, we have
		\begin{equation}
			\frac{r_2^{(e)}}{\mathcal{P}_2}=\frac{\prod_{j=1}^{R-1} (1+\varepsilon^{I_j\cup I_{j+1}})}{\prod_{j=2}^{R-1}(1+\varepsilon^{I_j})}\,,
		\end{equation}
		as so, using $|\varepsilon^{I}|\leq \varepsilon$,
		\begin{align}
			\frac{r_2^{(e)}}{\mathcal{P}_2}&\leq \frac{(1+\varepsilon )^{R-1}}{(1-\varepsilon)^{R-2}}\leq \left[ \frac{1+\varepsilon}{1-\varepsilon}\right]^{R}\,,\nonumber\\
			\frac{r_2^{(e)}}{\mathcal{P}_2}&\geq \frac{(1-\varepsilon )^{R-1}}{(1+\varepsilon)^{R-2}}\geq \left[ \frac{1-\varepsilon}{1+\varepsilon}\right]^{R}\,.
		\end{align}
		Since $\varepsilon\leq 1/2$, and using $(1+x)\leq e^{x}$, we have
		\begin{equation}
			\frac{1+\varepsilon}{1-\varepsilon}=1+\frac{2\varepsilon}{1-\varepsilon}\leq 1+4\varepsilon\leq e^{4\varepsilon}\,.
		\end{equation}
		Therefore, 
		\begin{equation}
			e^{-(4 L/k) \varepsilon}\leq \frac{r_2^{(e)}}{\mathcal{P}_2}\leq e^{(4 L/k) \varepsilon}\,.
		\end{equation}
		Finally, using 
		\begin{equation}
			(1-e^{-a})\leq e^{a}-1\,,
		\end{equation}
		we obtain~\eqref{eq:inequality_eps_appendix}
	\end{proof}
	\begin{lem}\label{lem:single_bound}
		Let $M\geq 2^{4|I|}$. Then
		\begin{equation}\label{eq:single_bound_1}
			{\rm Pr}\left(|\varepsilon^{I}|\geq \alpha\right)\leq \frac{2^{2|I|+3}}{\alpha^2 M}\,.
		\end{equation}	
	\end{lem}
	\begin{proof}
		We start from the bound~\eqref{eq:variance_purity}. Since $M\geq 2^{4|I|}$, we have in particular $M\geq 4$ and Eq.~\eqref{eq:variance_purity} implies
		\begin{align}
			{\rm Var}\left[\mathcal{P}_2^{(e)}[I]\right]&\leq 4\left(\frac{2^{|I|}\mathcal{P}_2[I]}{M}\right)+4\left(\frac{2^{2|I|}}{M}\right)^2\,.
		\end{align}
		Since by hypothesis $M\geq 2^{4|I|}$, while $\mathcal{P}_2[I]\geq 2^{-|I|}$, we have
		\begin{equation}
			4\left(\frac{2^{2|I|}}{M}\right)^2\leq 4\left(\frac{2^{|I|}\mathcal{P}_2[I]}{M}\right),
		\end{equation}
		and so 
		\begin{equation}\label{eq:relative_variance}
			{\rm Var}\left[\mathcal{P}_2^{(e)}[I]\right]\leq  8\frac{2^{|I|}\mathcal{P}_2[I]}{M} \,.
		\end{equation} 
		Eq.~\eqref{eq:single_bound_1} then follows using Chebyshev's inequality and $\mathcal{P}_2[I]\geq 2^{-|I|}$. 	
	\end{proof}
	
	We are ready to state and prove the main result of this section.
	\begin{thm}\label{th:main_theorem_1}
		Let $0\leq \delta\leq 1$ and take
		\begin{equation}
			M\geq {\rm max}\left\{2^{8k}, L^2\frac{2^{4k+10}}{k^2\delta^2}\right\}\,,
		\end{equation}
		where $k=|I_j|\geq 2$ (the case $k=1$ is trivial). Then	
		\begin{equation}\label{eq:main_result_appendix}
			{\rm Pr}\left[\left|(r_2^{(e)}/\mathcal{P}_2)-1\right|\geq \delta\right]\leq  \frac{2^{4k+11}L^3}{\delta^2k^3 M}\,.
		\end{equation}
	\end{thm}
	\begin{proof}
		Recalling the definition~\eqref{eq:def_varepsilon_appendix}, note that if $\varepsilon< (k/8 L)\delta$, then $\varepsilon<1/2$ and also $e^{(4L/k)\varepsilon}-1< \delta$. Using Lemma~\ref{lem:inequality_eps}, this implies
		\begin{align}\label{eq:almost_there}
			{\rm Pr}\left[\left|(r_2^{(e)}/\mathcal{P}_2)-1\right|\right.&\left.\geq \delta\right]\leq {\rm Pr}[ \varepsilon \geq (k/8 L)\delta]\nonumber\\
			&=1-{\rm Pr}[ \varepsilon < (k/8L)\delta]\,.
		\end{align}
		In more detail, the validity of the first inequality can be seen as follows: Suppose $\varepsilon<  (k/8L)\delta$. Then necessarily $\varepsilon< 1/2$ and so also $|r^{(e)}_2/\mathcal{P}_2-1|<\delta$ (by Lemma~\ref{lem:inequality_eps}). Therefore, the set of cases in which $|r^{(e)}_2/\mathcal{P}_2-1|\geq\delta$ must be contained in the set of cases in which $\varepsilon \geq  (k/8L)\delta$. In formula, this is the first line of Eq.~\eqref{eq:almost_there}. 
		
		From the definition of $\varepsilon$ and Lemma~\ref{lem:single_bound}, we have
		\begin{align}
			{\rm Pr}[ \varepsilon < x]&= \prod_j\left(1- {\rm Pr}[ |\varepsilon^{I_j}| \geq x]\right)\nonumber\\
			&\times  \prod_j\left(1- {\rm Pr}[|\varepsilon^{I_j\cup I_{j+1}}| \geq x]\right)\nonumber\\
			&\geq \left[\left(1-\frac{2^{2k+3}}{x^2M}\right)\left(1-\frac{2^{4k+3}}{x^2M}\right)\right]^{R}\,,
		\end{align}
		where we used that the random variables $\varepsilon^{I_j}$ are statistically independent. Therefore
		\begin{equation}
			{\rm Pr}[ \varepsilon < x]\geq \left(1-\frac{2^{4k+3}}{x^2M}\right)^{2R}\,.
		\end{equation}
		
		Setting $x=\delta k/(8L)$, we have by hypothesis $2^{4k+3}/x^2M\leq 1/2$. Therefore, using $1-z\geq e^{-2z}$ for $0\leq z\leq 1/2$, we obtain
		\begin{equation}
			{\rm Pr}[ \varepsilon < (k/8L)\delta]\geq \exp\left[ - \frac{2^{4k+11}L^3}{M\delta^2 k^3} \right]\,.
		\end{equation}
		Plugging this into~\eqref{eq:almost_there}, and using $1-e^{-z}\leq z$, we finally obtain~\eqref{eq:main_result_appendix}.
	\end{proof}
	
	\subsection{Statistical-error analysis for the PT moments}
	\label{sec:additive_error_PT}
	
	In this Appendix we prove Eq.~\eqref{eq:prob_inequality_PT}. Considering the same protocol and using the same notation as Sec.~\ref{sec:pt_moments}, we set 
	\begin{subequations}\label{eq:single_bound_1_PT}
		\begin{align}
			p_3^{(e)}[A_2B_1]&=p_3[A_2B_1]+\varepsilon^{A_2B1}\,,\\
			\mathcal{P}_3^{(e)}[A_2]&=\mathcal{P}_3[A_2](1+\varepsilon^{A_2})\,,\\
			\mathcal{P}_3^{(e)}[B_1]&=\mathcal{P}_3[B_1](1+\varepsilon^{B_1})\,.
		\end{align}
	\end{subequations}
	Note that $\varepsilon^{A_2B_1}$ is an additive error, while $\varepsilon^{A_2}$, $\varepsilon^{B_1}$ are relative errors. We also define
	\begin{equation}\label{eq:def_varepsilon_PT}
		\varepsilon=2^{2|A_2|+2|B_1|}{\rm max}\left\{ |\varepsilon^{A_2B_1}|,
		|\varepsilon^{A_2}|, |\varepsilon^{B_1}| \right\}\,.
	\end{equation} 
	In the following, we will set $|A_2|=|B_1|=k$, where $k\geq 2\ell-1$ with $\ell$ being the circuit depth. With this choice, Eq.~\eqref{eq:factorization_sn} holds. 
	
	For clarity, we organize the proof into lemmas.
	\begin{lem}\label{lem:inequality_eps_PT}
		Suppose $\varepsilon\leq 1$. Then
		\begin{equation}\label{eq:inequality_eps_PT}
			\left |	s_3^{(e)}[A_2B_1]-s_3[A_2B_1]\right|\leq 16 \varepsilon \,,
		\end{equation}
		where $s_n[A_2B_1]$ is defined in Eq.~\eqref{eq:factorization_sn}.
		\begin{proof}
			We start from
			\begin{align}
				s^{(e)}_3[A_2B_1]&=\frac{p_3[A_2B_1]+\varepsilon^{A_2B_1}}{\mathcal{P}_3[A_2](1+\varepsilon^{A_2})\mathcal{P}_3[B_1](1+\varepsilon^{B_1})}\,,
			\end{align}
			which gives us
			\begin{align}
				&|s^{(e)}_3[A_2B_1]-s_3[A_2B_1]|\leq\nonumber\\
				& \frac{|p_3[A_2B_1]|}{\mathcal{P}_3[A_2]\mathcal{P}_3[B_1]}\frac{|\varepsilon^{A_2}|+|\varepsilon^{B_1}|+|\varepsilon^{A_2}\varepsilon^{B_1}|}{(1+\varepsilon^{A_2})(1+\varepsilon^{B_1})}\nonumber\\
				&
				+\frac{|\varepsilon^{A_2B_1}|}{\mathcal{P}_3[A_2]\mathcal{P}_3[B_1](1+\varepsilon^{A_2})(1+\varepsilon^{B_1})}\,.
			\end{align}
			Denoting by $\lambda_j$ the eigenvalues of $\rho^{T_{A_2}}_{A_2B1}$, we have 
			\begin{equation}
				|p_3[A_2B_1]|\leq \sum_j|\lambda_j|^3\leq \sum_j\lambda_j^2=\mathcal{P}[A_2B_1]\leq 1\,,
			\end{equation}
			where we used $\lambda_j\in [-1/2, 1]$~\cite{rana2023negative} and ${\rm Tr}[(\rho_{AB}^{T_A})^2]={\rm Tr}[\rho_{AB}^2]$~\cite{elben2020mixed}. In addition, since $\varepsilon\leq 1$,  clearly $|\varepsilon^{A_2}|\leq 1/2$, $|\varepsilon^{B_1}|\leq 1/2$ and so also $|\varepsilon^{A_2}\varepsilon^{B_1}|<|\varepsilon^{A_2}|$. Putting everything together, and using $\mathcal{P}_3[I]\leq 2^{-2|I|}$, we finally arrive at Eq.~\eqref{eq:inequality_eps_PT}.
		\end{proof}
	\end{lem}
	
	\begin{lem}\label{lem:single_bound_PT}
		Let $M\geq 3\cdot 2^{8k}$. Then
		\begin{subequations}
			\label{eq:single_bound_PT}
			\begin{align}
				{\rm Pr}\left(|\varepsilon^{A_2}|\geq \alpha\right)&\leq 27\frac{2^{3k}}{\alpha^2 M}\,, \label{eq:single_bound_PT_1}\\
				{\rm Pr}\left(|\varepsilon^{B_1}|\geq \alpha\right)&\leq 27\frac{2^{3k}}{\alpha^2 M}\,, \label{eq:single_bound_PT_2}\\
				{\rm Pr}\left(|\varepsilon^{A_2B_1}|\geq \alpha\right)&\leq 27\frac{2^{2k}}{\alpha^2 M}\,. \label{eq:single_bound_PT_3}
			\end{align}	
			
		\end{subequations}
	\end{lem}
	\begin{proof}
		We start from the bound~\cite{rath2021quantum}
		\begin{align}
			\operatorname{Var}\left[p^{(e)}_3[AB]\right] &\leq 9 \frac{2^{|AB|}}{M} \operatorname{Tr}\left(\rho_{AB}^4\right)\nonumber\\
			+&18 \frac{2^{3 |AB|}}{(M-1)^2} p_2[AB]+6 \frac{2^{6 |AB|}}{(M-2)^3}\,.
		\end{align}
		Since $M\geq  3\cdot 2^{8k}$ (and $k\geq 1$), then
		\begin{equation}
			\frac{18}{(M-1)^2}\leq 	\frac{27}{M^2}\,,\quad \frac{6}{(M-2)^3}\leq 	\frac{9}{M^3}\,.
		\end{equation}
		Therefore
		\begin{align}
			\operatorname{Var}\left[p^{(e)}_3[AB]\right] &\leq 9 \frac{2^{|AB|}}{M} \operatorname{Tr}\left(\rho_{AB}^4\right)\nonumber\\
			+&27 \frac{2^{3 |AB|}}{M^2} p_2[AB]+9 \frac{2^{6 |AB|}}{M^3}\,.
		\end{align}
		Since $M\geq 3\cdot 2^{8k}$, we have
		\begin{align}
			9 \frac{2^{|AB|}}{M}\operatorname{Tr}\left(\rho_{AB}^4\right)&\geq 27 \frac{2^{3 |AB|}}{M^2} p_2[AB]\,,\nonumber\\
			9 \frac{2^{|AB|}}{M}\operatorname{Tr}\left(\rho_{AB}^4\right)&\geq 9 \frac{2^{6 |AB|}}{M^3}\,.
		\end{align}
		In the first line we have used that $p_2[I]={\rm Tr}[\rho_I^2]$ and H\"older's inequality, which guarantees 
		\begin{equation}
			{\rm Tr}[\rho^{2}_I]/{\rm Tr}[\rho^{4}_I]\leq 2^{|I|}/{\rm Tr[\rho_I^{2}]}\leq 2^{2|I|}\,,
		\end{equation}
		while in the second line we have used ${\rm Tr}(\rho_{AB}^{4})\geq 2^{-3|AB|}=2^{-6k}$. Putting all together, we get
		\begin{equation}
			\operatorname{Var}\left[p^{(e)}_3[AB]\right] \leq 27 \frac{2^{2k}}{M}\operatorname{Tr}\left(\rho_{AB}^4\right)\,.
		\end{equation}
		
		Similarly, we have~\cite{rath2021quantum}
		\begin{align}
			\operatorname{Var}\left[\mathcal{P}_3[A_2]\right] &\leq 9 \frac{2^{|A_2|}}{M} \operatorname{Tr}\left(\rho_{A_2}^4\right)\nonumber\\
			+&18 \frac{2^{3 |A_2|}}{(M-1)^2} \operatorname{Tr}\left(\rho_{A_2}^2\right)+6 \frac{2^{6 |A_2|}}{(M-2)^3}\nonumber\\
			\leq & 27 \frac{2^{k}}{M} \operatorname{Tr}\left(\rho_{A_2}^4\right)\,.
		\end{align}
		Using ${\rm Tr}\left(\rho_{A_2}^4\right)/({\rm Tr}(\rho_{A_2}^{3}))^2\leq 1/{\rm Tr}(\rho_{A_2}^{3})\leq 2^{2k}$, we obtain
		\begin{equation}
			{\rm Var}\left[\frac{\mathcal{P}^{(e)}_3[A_2]}{\mathcal{P}_3[A_2]}\right]\leq 27 \frac{2^{k}}{M} \frac{{\rm Tr}\left(\rho_{A_2}^4\right)}{({\rm Tr}(\rho_{A_2}^{3}))^2}\leq 27 \frac{2^{3k}}{M}\,,
		\end{equation}
		and analogously for $\mathcal{P}^{(2)}_3(B_1)$. Eqs.~\eqref{eq:single_bound_PT} follow using Chebyshev's inequality. 
	\end{proof}
	
	We are ready to state and prove the main result of this section, yielding Eq.~\eqref{eq:prob_inequality_PT}.
	\begin{thm}\label{th:main_theorem_PT}
		Let $0\leq \delta\leq 1$ and take
		\begin{equation}
			M\geq 27\frac{2^{11k+9}}{\delta^2}\,.
		\end{equation}
		Then	
		\begin{equation}\label{eq:main_result_PT}
			{\rm Pr}\left[\left|s^{(e)}_3[A_2B_1]-s_3[A_2B_1]\right|\geq \delta\right]\leq  81\frac{2^{11k+9}}{M\delta^2}\,.
		\end{equation}
	\end{thm}
	\begin{proof}
		Recall the definition~\eqref{eq:def_varepsilon_PT} and note that if $\varepsilon\leq \delta /16$, then trivially $\varepsilon<1$ and so, by Lemma~\ref{lem:inequality_eps_PT}, $|s^{(e)}_3-s_3|\leq \delta$. This implies
		\begin{align}\label{eq:almost_there_PT}
			{\rm Pr}[|s^{(e)}_3-s_3|&\geq \delta]\leq {\rm Pr}[ \varepsilon \geq \delta /16]\nonumber\\
			&=1-{\rm Pr}[ \varepsilon < \delta /16]\,.
		\end{align}
		In more detail, the validity of the first inequality can be seen as follows: Suppose $\varepsilon < \delta /16$. Then trivially $\varepsilon<1$ and so necessarily also $|s^{(e)}_3-s_3|<\delta$ (by Lemma~\ref{lem:inequality_eps_PT}). Therefore, the set of cases in which $|s^{(e)}_3-s_3|\geq \delta$ is contained in the set of cases for which $\varepsilon \geq \delta /16$. In formula, this is the first line of Eq.~\eqref{eq:almost_there_PT}.
		
		From the definition of $\varepsilon$ and Lemma~\ref{lem:single_bound_PT}, we have
		\begin{align}
			& {\rm Pr}[ \varepsilon < x]= (1-{\rm Pr}[|\varepsilon^{A_2}|\geq x 2^{-4k}])\nonumber\\
			&\times (1-{\rm Pr}[|\varepsilon^{B_1}|\geq x 2^{-4k}]) (1-{\rm Pr}[|\varepsilon^{A_2B_1}|\geq x2^{-4k} ])\,\nonumber\\
			&\geq\left(1-27\frac{2^{11k}}{x^2 M}\right)^3\,,
		\end{align}
		where we used that the random variables $\varepsilon^{I_j}$ are statistically independent. 
		Setting $x=\delta /16$, we have by hypothesis $27\cdot 2^{11k}/x^2M<1/2$. Therefore, using $1-z\geq e^{-2z}$ for $0\leq z\leq 1/2$, we obtain
		\begin{equation}
			{\rm Pr}[ \varepsilon < \delta /16]\geq \exp\left[ - 81\frac{2^{11k+9}}{M\delta^2} \right]\,.
		\end{equation}
		Plugging into~\eqref{eq:almost_there_PT}, and using $1-e^{-z}\leq z$, we finally obtain~\eqref{eq:main_result_PT}.
	\end{proof}
	
	\section{Statistical-error analysis for states satisfying the AFCs}
	\label{sec:purity_short_range}
	
	In this Appendix we prove the main claims reported in Sec.~\ref{sec:purity_PT_estimation_AFCs}. Let $\rho$ be a state satisfying~\eqref{eq:approx_split_1} for all partitions as in Fig.~\ref{fig:partition}$(a)$, with $|B|=k\geq k_c$ (and open boundary conditions). We first prove Eq.~\eqref{eq:final_result_purity_short_range}. 
	
	We start with the following:
	\begin{lem}\label{lem:inequality_short_range}
		Let 
		\begin{equation}\label{eq:appendix_inquality_k}
			k\geq \frac{\log(\alpha_2 L/\delta)}{\beta_2}\,,
		\end{equation}
		for $0\leq \delta \leq 1/2$. Then, defining $r_2$ as in Eq.~\eqref{eq:def_r2}, we have
		\begin{equation}\label{eq:main_result_appendix_short_range}
			\left|\frac{r_2}{{\rm Tr}(\rho^2)}-1\right|\leq 4\delta\,.
		\end{equation}
	\end{lem}
	\begin{proof}
		First, applying iteratively~\eqref{eq:approx_split_1}, we get (recalling that $R=L/k$)
		\begin{equation}
			\frac{r_2}{{\rm Tr}(\rho^2)}=\prod_{j=2}^{R-1}(1+\varepsilon_j)\,,
		\end{equation}
		with $|\varepsilon_j|\leq \alpha_2 e^{-\beta_2 k}$, and so
		\begin{equation}
			\left|\frac{r_2}{{\rm Tr}(\rho^2)}-1\right|=\left|\prod_{j=2}^{R-1}(1+\varepsilon_j)-1\right|\,.
		\end{equation}
		We have
		\begin{align}
			&\prod_{j=2}^{R-1}(1+\varepsilon_j)-1\leq (1+\alpha_2 e^{-\beta_2k})^{R-2}-1\nonumber\\
			&\leq \exp\left[\alpha_2 (L/k) e^{-\beta_2k}\right]-1\,.
		\end{align}
		Analogously, 
		\begin{equation}
			\prod_{j=2}^{R-1}(1+\varepsilon_j)-1\geq (1-\alpha_2 e^{-\beta_2k})^{R-2}-1\,.
		\end{equation}
		Due to~\eqref{eq:appendix_inquality_k}, we have $\alpha_2 e^{-\beta_2k}\leq 1/2$. Therefore, using $1-z\geq e^{-2z}$ for $0\leq z\leq 1/2$, we have
		\begin{equation}
			\prod_{j=2}^{R-1}(1+\varepsilon_j)-1\geq \exp\left[-2\alpha_2 (L/k) e^{-\beta_2k}\right]-1\,.
		\end{equation}
		Therefore
		\begin{equation}
			e^{-y}-1\leq \prod_{j=2}^{R-1}(1+\varepsilon_j)-1\leq e^{y/2}-1\leq e^{y}-1\,,
		\end{equation}
		where $y=2\alpha_2 (L/k) e^{-\beta_2k}$. Since $1-e^{-y}\leq e^{y}-1$, this implies
		\begin{equation}
			\left| \prod_{j=2}^{R-1}(1+\varepsilon_j)-1\right|\leq \exp\left[2\alpha_2 (L/k) e^{-\beta_2k}\right]-1\,.
		\end{equation}
		Finally, due to~\eqref{eq:appendix_inquality_k}, we have $2\alpha_2 (L/k) e^{-\beta_2k}\leq 1$. Therefore, using $e^{z}-1\leq 2z$ for $0\leq z\leq 1$, we arrive at
		\begin{equation}
			\left| \frac{r_2}{{\rm Tr}(\rho^2)} -1\right|\leq  \frac{4\alpha_2L}{k}e^{-\beta_2 k}\leq 4\delta\,.
		\end{equation}
	\end{proof}
	
	Next, we prove Eq.~\eqref{eq:final_result_purity_short_range}, via the following:
	\begin{thm}
		Let $0\leq \delta \leq 1/2$ and set
		\begin{equation}
			k= \frac{\log(\alpha_2 L/\delta)}{\beta_2}\,.
		\end{equation}
		Choosing
		\begin{equation}\label{eq:inequality}
			M\geq {\rm max}\left\{
			\left(\frac{\alpha_2 L}{\delta}\right)^{\frac{8\log 2}{\beta_2}}, \frac{L^22^{10}}{\delta^2}\left(\frac{\alpha_2 L}{\delta}\right)^{\frac{4\log 2}{\beta_2}}
			\right\}\,,
		\end{equation}
		and recalling the definition~\eqref{eq:r2_e}, we have
		\begin{align}
			&{\rm Pr}\left[\left|(r_2^{(e)}/\mathcal{P}_2)-1\right|\geq 7\delta\right]\leq \frac{2^{11}L^3}{\delta^2 M}\left(\frac{\alpha_2 L}{\delta}\right)^{\frac{4\log 2}{\beta_2}}\,.
		\end{align}
	\end{thm}
	\begin{proof}
		Set $r_2^{(e)}/r_2=(1+\varepsilon_1)$ and $r_2/\mathcal{P}_2=(1+\varepsilon_2)$, where $r^{(e)}_2$ is defined in Eq.~\eqref{eq:r2_e}. Using Lemma~\ref{lem:inequality_short_range} (and that $\delta\leq 1/2$),  we have $|\varepsilon_2|\leq 4\delta$ and so
		\begin{align}
			\left|\frac{r_2^{(e)}}{\mathcal{P}_2}-1\right|&= \left|\frac{r_2^{(e)}}{r_2}\frac{r_2}{\mathcal{P}_2}-1\right| \leq 4\delta + 3 |\varepsilon_1| \,.
		\end{align}
		If $|\varepsilon_1|< \delta$, then $|r^{(e)}_2/\mathcal{P}_2-1|< 7 \delta$. Therefore, the set of cases in which $|r^{(e)}_2/\mathcal{P}_2-1|\geq 7 \delta$ is contained in the set of cases in which $|\varepsilon_1|\geq \delta$. In formula, 
		\begin{align}
			{\rm Pr}[|(r_2^{(e)}/\mathcal{P}_2)-1|&\geq 7\delta]\leq {\rm Pr}[|(r_2^{(e)}/r_2)-1|\geq \delta|]\,.
		\end{align}
		Thanks to Eq.~\eqref{eq:inequality}, we can use Theorem~\ref{th:main_theorem_1}, yielding
		\begin{align}
			{\rm Pr}[|(r_2^{(e)}/r_2)-1|\geq \delta|]&\leq  \frac{2^{4k+11}L^3}{\delta^2k^3 M}\leq  \frac{2^{4k+11}L^3}{\delta^2M}\nonumber\\
			&\leq  \frac{2^{11}L^3}{\delta^2 M}\left(\frac{\alpha_2 L}{\delta}\right)^{\frac{4\log 2}{\beta_2}}\,,
		\end{align}
		which completes the proof.
	\end{proof}
	\noindent Note that, in order to recover Eqs.~\eqref{eq:global_approximation_condition} and ~\eqref{eq:final_result_purity_short_range}, we simply rename $\delta'=7\delta$.
	
	Finally, we present a technical result showing that the PT moments can be estimated efficiently, assuming that the state to be measured satisfies the AFCs. We focus for simplicity on the case $n=3$ and follow a protocol similar to that of Sec.~\ref{sec:pt_moments}. Namely, we take $s^{(e)}_3[A_2B_1]$, defined in~\eqref{eq:estimator_s}, as our estimator for $\tilde{p}_3[AB]$, and, for any $0\leq \delta \leq 1$, we choose
	\begin{equation}
		k= \frac{\log(2\alpha_3 /\delta)}{2\beta_3}\,.
	\end{equation}
	Finally, we perform $M$ measurements to estimate $p_3[A_2B_1]$, $\mathcal{P}_3[A_2]$, and $\mathcal{P}_3[B_1]$ each (so the total number is $M_T=3M$). 
	We can then prove the following:
	\begin{thm}
		For any $0\leq \delta\leq 1/2$, set
		\begin{equation}\label{eq:k_dependence}
			k= \frac{\log(2\alpha_3/\delta)}{2\beta_3}\,.
		\end{equation}
		If
		\begin{equation}\label{eq:inequality_PT}
			M\geq 27\frac{2^{9}}{\delta^2}\left(\frac{2\alpha_3}{\delta}\right)^{\frac{11\log 2}{2\beta_3}}\,,
		\end{equation}
		then
		\begin{align}
			{\rm Pr}\left[\left|s^{(e)}_3-\tilde{p}_3\right|\geq \delta\right] \leq \frac{81\cdot 2^{9}}{\delta^2 M}\left(\frac{2\alpha_3}{\delta}\right)^{\frac{11\log 2}{2\beta_3}}\,.
		\end{align}
	\end{thm}
	\begin{proof}
		Our estimator for the normalized PT moment $\tilde{p}_3[AB]$ is
		\begin{equation}
			s^{(e)}_3=\frac{p^{(e)}_3[A_2B_1]}{\mathcal{P}^{(e)}_3[A_2]\mathcal{P}^{(e)}_3[B_1]}\,.
		\end{equation}
		First, suppose $|s_3-s^{(e)}_3|<\delta/2$. Then, using Eqs.~\eqref{eq:approx_split_2} and~\eqref{eq:k_dependence}, we have 
		\begin{align}
			|s^{(e)}_3-\tilde{p}_3|&\leq |s^{(e)}_3-s_3|+|s_3-\tilde{p}_3|\nonumber\\
			&\leq (\delta/2)+(\delta/2)=\delta\,.
		\end{align}
		Therefore, the set of cases in which $|s^{(e)}_3-\tilde{p}_3|\geq \delta$ is contained in the set of cases in which $|s^{(e)}_3-s_3|\geq \delta/2$, namely
		\begin{align}
			{\rm Pr}[|s^{(e)}_3-\tilde{p}_3|\geq \delta]&\leq {\rm Pr}[|s^{(e)}_3-s_3|\geq \delta] \,.
		\end{align}
		Finally, thanks to Eq.~\eqref{eq:inequality_PT}, we can apply Theorem~\ref{th:main_theorem_PT}, yielding
		\begin{align}
			{\rm Pr}\left[\left|s^{(e)}_3-\tilde{p}_3\right|\geq \delta\right] \leq \frac{81\cdot 2^{9}}{\delta^2 M}\left(\frac{2\alpha_3}{\delta}\right)^{\frac{11\log 2}{2\beta_3}}\,.
		\end{align}
		
	\end{proof}
	
	\section{AFCs and MPDOs}
	\label{sec:factorization_MPDOs}
	
	In this Appendix we prove that the purity AFCs hold for MPDOs. To this end, we assume the following conditions on the transfer matrices~\eqref{eq:tau_n}:
	\begin{enumerate}[label=(\Alph*)]
		\item The matrices $\tau_1$ and $\tau_2$ admit the spectral decomposition
		\begin{subequations}\label{eq:transfer_matrices}
			\begin{align}
				\tau_1&=\sum_{j=0}^{\chi-1}\lambda_j |R^{(1)}_j\rangle\langle L^{(1)}_j|\,,\\
				\tau_2&=\sum_{j=0}^{\chi^2-1}\mu_j |R^{(2)}_j\rangle\langle L^{(2)}_j|\,,
			\end{align}
		\end{subequations}
		where we assume $|\lambda_0|>|\lambda_j|$, $|\mu_0|>|\mu_j|$ for all $j\geq 1$, \emph{i.e.}~$\tau_1$, $\tau_2$ have a trivial Jordan form and a finite gap.\footnote{The assumption that $\tau_1$ and $\tau_2$ can be diagonalized is purely technical and not necessary. However, we keep it here as it makes the analysis simpler and it is in any case quite general.} Note that we can also assume without the loss of generality that $\lambda_0=1$ (which implies $\mu_0>0$, since $\rho_L>0$ for all $L$). Finally, $\ket{R^{(n)}_j}$ and $\bra{L_j^{(n)}}$ are the left and right eigenstates, \emph{i.e}, they are vectors on the left/right virtual indices of $\tau_1$, and $\tau_2$ which are normalized such that 
		\begin{equation}
			\braket{L^{(n)}_j| R^{(n)}_k}=\delta_{j,k}\,\quad n=1,2.
		\end{equation}
		\item  We further need to assume the technical conditions
		\begin{subequations}
			\label{eq:eigenvectors_conditions}			
			\begin{align}
				(\langle L^{(1)}_0 |\otimes \langle L^{(1)}_0 |) | R^{(2)}_0\rangle&\neq 0\,, \\
				\langle L^{(2)}_0|(| R^{(1)}_0 \rangle\otimes| R^{(1)}_0 \rangle)&\neq 0\,,
			\end{align}
			which again are quite general, as orthogonality requires fine-tuning.
		\end{subequations}
	\end{enumerate}
	
	We use the same notations and assumptions as in Sec.~\ref{sec:MPDOs}, so that
	\begin{equation}\label{eq:p2_sigma_appendix}
		\mathcal{P}_2=\frac{{\rm Tr}[\sigma^2]}{({\rm Tr}[\sigma])^{2}}\,,
	\end{equation}
	and
	\begin{equation}
		r_2=\frac{\prod_{j=1}^{R}{\rm Tr}_{I_j\cup I_{j+1}}(\sigma^2_{I_{j}\cup I_{j+1}})}{\prod_{j=1}^{R}{\rm Tr}_{I_j}\left[\sigma_{I_j}^{2}\right]}\,,
	\end{equation}
	where $|I_j|=k$ and $R=L/k$ is an integer.

	\begin{figure*}[t]
		\includegraphics[width=0.95\textwidth]{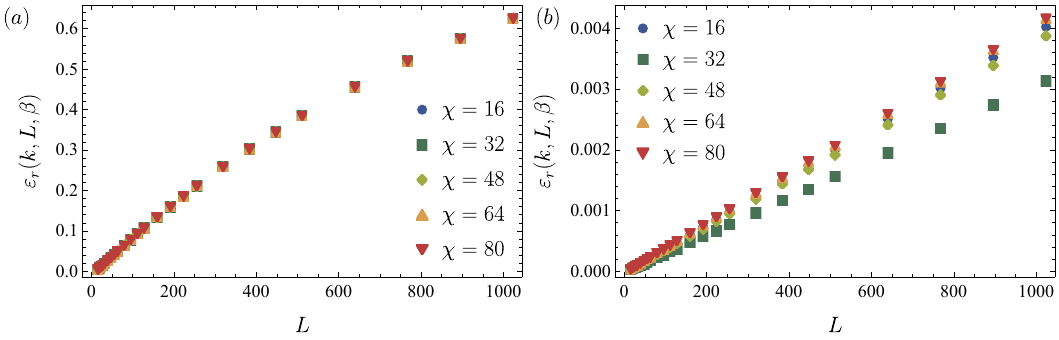}
		\caption{Scaling of the relative error $\varepsilon_r(k,L,\beta)$, defined in Eq.~\eqref{eq:relative_error_numerics}, as a function of the bond dimension $\chi$ used to approximate the thermal state. In the left panel we show data for the Ising chain while the right panel shows data for the XXZ chain. In both cases we use the same parameters as in Fig.~\ref{fig:decayingK}. Data shown for $k=4$, other values of $k$ show similar behavior with $\chi=16$ being sufficient for convergence in the Ising model, while the XXZ model appears to be well converged from $\chi=64$. }
		\label{fig:convergence_bond_dimension}
	\end{figure*}

	First, we introduce the correlation lengths
	\begin{equation}
		e^{-\frac{1}{\zeta_1}}=\frac{ {\rm max}_{j>0}\{|\lambda_j|\}}{\lambda_0}\,,\quad  e^{-\frac{1}{\zeta_2}}=\frac{ {\rm max}_{j>0}\{|\mu_j|\}}{\mu_0}\,,
	\end{equation}
	and also
	\begin{equation}
		\zeta={\rm max}(\zeta_1,\zeta_2)\,.
	\end{equation}
	Next, given the spectral decomposition in Eqs.~\eqref{eq:transfer_matrices}, we define
	\begin{equation}\label{eq:c_constant}
		C=\underset{\substack{(j,k,l)\\ \neq (0,0,0)}}{{\rm max}}\left\{ \frac{| \langle L_j^{(1)}| \langle L_k^{(1)} |R_l^{(2)}\rangle\langle L_l^{(2)}|  R_j^{(1)}\rangle | R_k^{(1)} \rangle| }{| \langle L_0^{(1)}| \langle L_0^{(1)} |R_0^{(2)}\rangle\langle L_0^{(2)}|  R_0^{(1)}\rangle | R_0^{(1)} \rangle|}\right\}\,.
	\end{equation}
	Note that the denominator is non-vanishing because of Eqs.~\eqref{eq:eigenvectors_conditions}. We can now state the main result of this section.
	\begin{thm}\label{th:MPDO_approximation}
		Take $|I_j|=k$ with
		\begin{equation}\label{eq:kc}
			k\geq k_{\rm min}={\rm max}\left\{1, \zeta \log (20 C \chi^2 L)\right\}\,.
		\end{equation}
		Then, for all
		\begin{equation}
			L\geq {\rm max}\{\zeta\log(2^{5}\chi^2), 4k +\zeta\log (2\chi)\}\,,
		\end{equation} 
		we have
		\begin{equation}
			\left| \frac{r_2}{\mathcal{P}_2}-1  \right| \leq \chi^2 (80C+32) \frac{L}{k} e^{-k /\zeta}\,.
		\end{equation}
	\end{thm}
	\begin{proof}
		Using Eqs.~\eqref{eq:transfer_matrices}, we have
		\begin{equation}\label{eq:p2_frac_eigen}
			\mathcal{P}_2=\frac{\sum_{j=0}^{\chi^{2}-1}\mu_j^L}{\left(\sum_{j=0}^{\chi-1}\lambda_j^L\right)^2}= \frac{\mu_0^L}{\lambda_0^{2L}}\left(1+\tilde{\varepsilon}\right) \,,
		\end{equation}
		where
		\begin{equation}\label{eq:tilde_epsilon}
			\tilde{\varepsilon}= \frac{1+\sum_{j=1}^{\chi^2-1}(\mu_j/\mu_0)^{L}}{\left(1+\sum_{j=1}^{\chi-1}(\lambda_j/\lambda_0)^{L}\right)^2}-1\,.
		\end{equation}
		Since by hypothesis $L\geq\zeta_1 \log(2\chi)$, we have $\chi e^{-L/\zeta_1}\leq 1/2$, and so
		\begin{align}\label{eq:temp_0}
			|\tilde{\varepsilon}|&\leq 4\left(\chi^2 e^{-L/\zeta_2 }+2\chi e^{-L/\zeta_1}+\chi^2e^{-2L/\zeta_1 }\right)\nonumber\\
			&\leq 2^{4}\chi^2 e^{-L/\zeta}\,.
		\end{align}
		On the other hand
		\begin{align}
			&{\rm Tr}[\sigma_{I}^2]=\sum_{j,k=0}^{\chi-1}\lambda^{L-|I|}_j\lambda^{L-|I|}_k(\langle L_j^{(1)}|\otimes \langle L_k^{(1)}|)\nonumber\\
			&\left(\sum_{l=0}^{\chi^2-1}\mu^{|I|}_l |R_l^{(2)}\rangle\langle L_l^{(2)}|\right) (|R_j^{(1)}\rangle\otimes | R_k^{(1)} \rangle)\,.
		\end{align}
		Therefore, recalling the definition~\eqref{eq:c_constant}, we have
		\begin{align}
			{\rm Tr}&[\sigma_{I}^2]=\lambda_0^{2(L-|I|)}\mu_0^{|I|}\langle L_0^{(1)}| \langle L_0^{(1)} |R_0^{(2)}\rangle\nonumber\\
			&\times \langle L_0^{(2)}|  R_0^{(1)}\rangle | R_0^{(1)} \rangle\left(1+\varepsilon^{I}\right)\,,
		\end{align}
		where 
		\begin{align}
			|\varepsilon^{I}|&\leq 2 C\chi e^{-(L-|I|)/\zeta_1 }+  C\chi^2 e^{-|I|/\zeta_2 }\nonumber\\
			+&  C\chi^2 e^{-2(L-|I|)/\zeta_1 }+ 2C \chi^3 e^{-(L-|I|)/\zeta_1 }e^{-|I|/\zeta_2 }\nonumber\\
			+& C\chi^4 e^{-2(L-|I|)/\zeta_1 }e^{-|I|/\zeta_2 }\,.
		\end{align}
		Since by hypothesis $L\geq 4k +\zeta\log (2\chi)$,  it is easy to verify that all the five terms above are upper bounded by $C\chi^2 e^{-|I|/\zeta }$ (recall that $|I|\leq 2k$ for all $I$), and so
		\begin{align}\label{eq:inequality_eps_I}
			|\varepsilon^{I}|&\leq 5 C\chi^2 e^{-|I|/\zeta}\leq 5 C\chi^2 e^{-k/\zeta}\,,
		\end{align}
		where we used that either $|I|=k$ or $|I|=2k$, and so $|I|\geq k$. Therefore
		\begin{align}
			&\frac{\prod_{j=1}^{R}{\rm Tr}_{I_j\cup I_{j+1}}(\sigma^2_{I_{j}\cup I_{j+1}})}{\prod_{j=1}^{R}{\rm Tr}_{I_j}\left[\sigma_{I_j}^{2}\right]}\nonumber\\
			&=\frac{\mu_0^L}{\lambda_0^{2L}}\left(\frac{\prod_{j=1}^R(1+\varepsilon^{I_j\cup I_{j+1}})}{\prod_{j=1}^R(1+\varepsilon^{I_j})}\right)\,,
		\end{align}
		with $	|\varepsilon^{I}|\leq 5C\chi^2 e^{-k/\zeta}$. Combining this with Eq.~\eqref{eq:p2_frac_eigen}, we arrive at
		\begin{align}\label{eq:temp_1}
			\frac{r_2}{\mathcal{P}_2}-1&=\left(\frac{\prod_{j=1}^R(1+\varepsilon^{I_j\cup I_{j+1}})}{\prod_{j=1}^R(1+\varepsilon^{I_j})}\right)\left(1+\tilde{\varepsilon}\right)^{-1}-1\nonumber\\
			=&\left(1+\tilde{\varepsilon}\right)^{-1}\left[
			\frac{\prod_{j=1}^R(1+\varepsilon^{I_j\cup I_{j+1}})}{\prod_{j=1}^R(1+\varepsilon^{I_j})}-1-\tilde{\varepsilon}
			\right]\,,
		\end{align}
		where $\tilde{\varepsilon}$ is given in Eq.~\eqref{eq:tilde_epsilon}.  
		
		Next, we define
		\begin{equation}\label{eq:def_varepsilon_periodic}
			\varepsilon={\rm max}\left\{|\varepsilon^{I}|: I\in \{I_{j}\}_j\cup\{I_{j}\cup I_{j+1}\}_{j}\right\}\,.
		\end{equation} 
		By hypothesis $k\geq k_{\rm min}$ [$k_{\rm min}$ is given in~\eqref{eq:kc}]. Therefore, using~\eqref{eq:inequality_eps_I} we have $|\varepsilon^{I}|\leq 1/2$ for all $|I|$ and so also $\varepsilon\leq 1/2$. Therefore, we can apply the derivation in Lemma~\ref{lem:inequality_eps} to show
		\begin{equation}
			\left|\frac{\prod_{j=1}^R(1+\varepsilon^{I_j\cup I_{j+1}})}{\prod_{j=1}^R(1+\varepsilon^{I_j})}-1\right|\leq e^{(4L/k) \varepsilon}-1\,.
		\end{equation}
		Since by hypothesis $L\geq \xi \log(2^{5}\chi^2)$, we also have $\tilde{\varepsilon}\leq 1/2$ [cf. Eq.~\eqref{eq:temp_0}], and so~\eqref{eq:temp_1} yields
		\begin{align}
			\left| \frac{r_2}{\mathcal{P}_2}-1  \right| &\leq  2(e^{(4L/k) \varepsilon}-1)+2|\tilde{\varepsilon}|\,.
		\end{align}
		Finally, we note that Eq.~\eqref{eq:kc} implies that $(4L/k) \varepsilon\leq 1$, and using $e^{z}-1\leq 2z$ for $0\leq z\leq 1$, we  arrive at
		\begin{align}
			\left| \frac{r_2}{\mathcal{P}_2}-1  \right| &\leq 80\chi^2 C (L/k) e^{-k /\zeta}+2^{5}\chi^2 e^{- L/\zeta}\nonumber\\
			&\leq \chi^2 (80C+32) \frac{L}{k} e^{-k /\zeta}\,.
		\end{align}
	\end{proof}
	
	This theorem proves Eq.~\eqref{eq:MPDO_approximation}, under the condition~\eqref{eq:final_kc}. Therefore, the approximate factorization property~\eqref{eq:global_approximation} holds for MPDOs, with the identification
	\begin{equation}
		\alpha_2=(20C+8)\chi^2\,,\quad \beta_2=1/\zeta\,.
	\end{equation}

	\section{Details on the numerical computations}
	\label{sec:appendix_numerics}
	
	In this Appendix we provide further details on the numerical computations performed to obtain the data presented in Secs.~\ref{sec:AFCs_numerics} and~\ref{sec:efficient_entanglement_detection}.  As mentioned, the calculations are carried out using the iTensor library~\cite{fishman2022itensor}, by first approximating the thermal states by MPOs of bond dimension $\chi$ and subsequently taking powers and traces of the density matrices represented in this way. For each quantity, we have always verified that the results were stable upon increasing the bond dimension $\chi$. An example of our data is reported in Fig.~\ref{fig:convergence_bond_dimension}, where we study $\varepsilon_r(k,L,\beta)$, defined in Eq.~\eqref{eq:relative_error_numerics}, as a function of the bond dimension $\chi$ used to approximate the thermal state. In general, we have found that relatively small bond dimensions are enough in the quantum Ising chain, while larger bond dimensions are required in order to observe convergence in the Heseinberg model.
	
	\begin{figure*}[t]
		\includegraphics[scale=0.66]{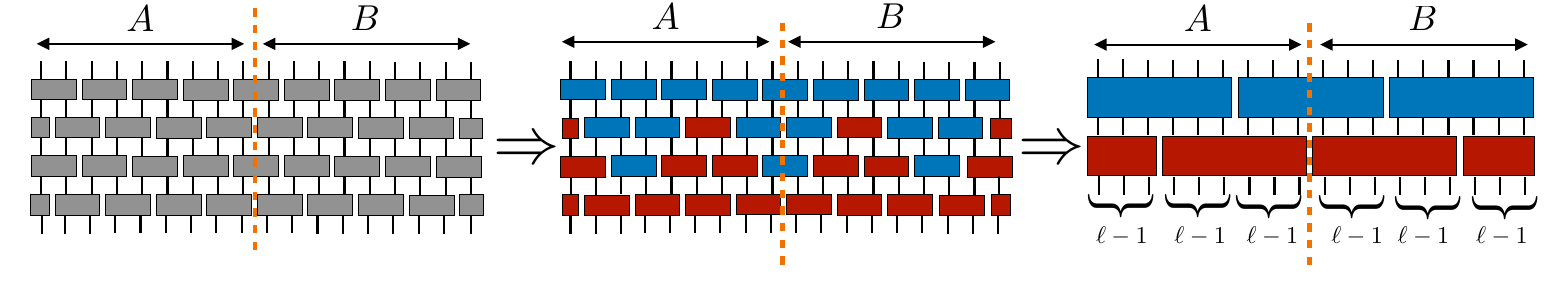}
		\caption{Any FDQC of depth $\ell$ can be rewritten as a depth-$2$ quantum circuit after grouping $\ell-1$ qubits into a single qudit. In the middle picture, we highlight with different colors sets of gates defining the two-qudit gates in the grouped lattice (colored rectangles in the right figure). In all the figures, a dashed orange line separates the regions $A$ and $B$, defining the bipartition of the system.}
		\label{fig:decomposition_circuit}
	\end{figure*}
	
	\section{Technical details on the PPT conditions}
	\label{sec:ppt_condition_for_short_range}
	
	The goal of this Appendix is to identify the conditions under which  the class of states $\rho^{(L)}(\gamma)$ introduced in Sec.~\ref{sec:examples} satisfy Eqs.~\eqref{eq:negativity_conditions} for all system sizes $L$ and some suitable constants $C_3$, $C_5$ independent of $L$. 
	
	For concreteness, let the circuit depth $\ell$ be even, $\ell=2k$ with $k\geq 1$ (a similar discussion holds for $\ell=2k+1$) and take a bipartition of the system as in Fig.~\ref{fig:decomposition_circuit}. We group neighboring sites into blocks containing $q=\ell-1$ qubits, forming a new super qudit associated with a Hilbert-space of dimension $d=2^{\ell-1}$. It is easy to show that the circuit can be rewritten as a depth-$2$ quantum circuit acting on the super qudits, cf. Fig.~\ref{fig:decomposition_circuit}. Therefore (assuming without loss of generality $R=L/q$ is an integer)
	\begin{equation}
		\rho_L(\gamma)=V^{(2)}\left(\bigotimes_{j=1}^{L/q}\omega_j\right)\left[V^{(2)}\right]^\dagger\,,
	\end{equation}
	where $\omega_j=\otimes_{i=1}^q \sigma_{jq+i}$, while
	\begin{equation}
		V^{(2)}=\left(\prod_{j}V_{2j,2j+1}\right)\left(\prod_{j}V_{2j+1,2j+2}\right)\,,
	\end{equation} 
	is the depth-$2$ FDQC acting on the super lattice. 
	
	As it is manifest from Fig.~\ref{fig:decomposition_circuit}, the region $A$ contains all super qudits with labels from $1$ to $L/2q$, while $B$ contains all those with labels from $L/2q+1$ to $R=L/q$. Define the sets of super qudits $S_1=\{R/2-1,R/2\}$, $S_2=\{R/2+1, R/2+2\}$ and
	\begin{equation}
		\tilde{\rho}_{S_1S_2}=W\omega_{R/2-1}\otimes \omega_{R/2}\otimes \omega_{R/2+1}\otimes\omega_{R/2+2}W^{\dagger}\,,
	\end{equation}
	where $W=V_{R/2,R/2+1}V_{R/2-1,R/2}V_{R/2+1,R/2+2}$. Note that $\tilde{\rho}_{S_1S_2}$ is different from the reduced density matrix over $S_1\cup S_2$. Using the graphical representation for the blocked circuits and Eq.~\eqref{eq:factorization_sn}, it is simple to show
	\begin{equation}
		\tilde{p}_n[AB]=\frac{{\rm Tr}\left[\left(\tilde{\rho}_{S_1S_2}^{T_{S_1}}\right)^n\right]}{{\rm Tr}_{S_1}[\tilde{\rho}_{S_1}^n]{\rm Tr}_{S_2}[\tilde{\rho}_{S_2}^n]}=:\tilde{s}_n\,,
	\end{equation}
	where $\tilde{\rho}_{S_1}={\rm Tr}_{S_2}[\tilde{\rho}_{S_1S_2}]$ and $\tilde{\rho}_{S_2}={\rm Tr}_{S_1}[\tilde{\rho}_{S_1S_2}]$. Therefore, $\tilde{p}_n[AB]$ coincides with the normalized PT moments of the state $\tilde{\rho}_{S_1S_2}$, supported on four super qudits. Next, it is also easy to compute
	\begin{align}
		\mathcal{P}_n[A]&={\rm Tr}_{S_1}[\tilde{\rho}^n_{S_1}] [\gamma^n+(1-\gamma)^n]^{L/2-2q}\,, \nonumber\\
		\mathcal{P}_n[B]&={\rm Tr}_{S_2}[\tilde{\rho}^n_{S_2}] [\gamma^n+(1-\gamma)^n]^{L/2-2q}\,.
	\end{align}
	Finally, setting
	\begin{align}
		\tilde{t}_3&=	\frac{{\rm Tr}_{S_1}[\tilde{\rho}^2_{S_1}]^{2}{\rm Tr}_{S_2}[\tilde{\rho}^2_{S_2}]^{2}}{{\rm Tr}_{S_1}[\tilde{\rho}^3_{S_1}]{\rm Tr}_{S_2}[\tilde{\rho}^3_{S_2}]}\,,\\
		\tilde{t}_5&=	\frac{{\rm Tr}_{S_1}[\tilde{\rho}^4_{S_1}]^{2}{\rm Tr}_{S_2}[\tilde{\rho}^4_{S_2}]^{2}}{{\rm Tr}_{S_1}[\tilde{\rho}^3_{S_1}]{\rm Tr}_{S_2}[\tilde{\rho}^3_{S_2}]{\rm Tr}_{S_1}[\tilde{\rho}^5_{S_1}]{\rm Tr}_{S_2}[\tilde{\rho}^5_{S_2}]}\,,
	\end{align}
	we arrive at
	\begin{subequations}
		\begin{align}
			f_3&=\tilde{s}_3-\tilde{s}_2^2\tilde{t}_3\frac{[\gamma^2+(1-\gamma)^2]^{2L-8q}}{[\gamma^3+(1-\gamma)^3]^{L-4q}}\,,\label{eq:f3_eq}\\
			f_5&=\tilde{s}_5\tilde{s}_3-\tilde{s}_4^2\tilde{t}_5\frac{[\gamma^4+(1-\gamma)^4]^{2L-8q}}{[\gamma^3+(1-\gamma)^3]^{L-4q}[\gamma^5+(1-\gamma)^5]^{L-4q}}\,.\label{eq:f5_eq}
		\end{align}
		
	\end{subequations}
	
	Now, define
	\begin{align}
		K_3&={\rm max}\{\tilde s_3(\gamma)- \tilde{s}^2_2(\gamma) {\tilde t_3(\gamma)}: 0\leq \gamma \leq 1/4\}\,,\\
		K_5&={\rm max}\{\tilde s_5(\gamma)\tilde s_3(\gamma)- \tilde{s}^2_4(\gamma) {\tilde t_5(\gamma)}:  0\leq \gamma \leq 1/4\}\,,
	\end{align}
	and
	\begin{align}
		H_3&={\rm max}\{|\tilde{s}^2_2(\gamma) {\tilde t_3(\gamma)|}: 0\leq \gamma \leq 1/4\}\,,\\
		H_5&={\rm max}\{|\tilde{s}^2_4(\gamma) {\tilde t_5(\gamma)|}:  0\leq \gamma \leq 1/4\}\,.
	\end{align}
	The constants $K_3$, $K_5$, $H_3$, and $H_5$ depend on the specific choices of the unitary gates forming $U^{(\ell)}$. For non-entangling gates, $K_3$ and $K_5$ are positive, but for generic choices of gates one has $K_3<0$ and $K_5<0$. We are ready to state our main result.
	\begin{thm}
		Suppose $K_3$, $K_5< 0$. Then, assuming without loss of generality $\gamma\leq 1/4$ and defining
		\begin{equation}
			C_l=\frac{|K_l|}{2}\,, \quad \gamma_3=-\frac{K_3}{4H_3},\quad \gamma_5=\left(-\frac{K_5}{8H_5}\right)^{1/3}\,,
		\end{equation}
		we have
		\begin{equation}
			0\leq \gamma\leq \gamma_3/L \Rightarrow	f_3\leq -C_3\,,
		\end{equation}
		while
		\begin{equation}
			0\leq \gamma\leq \gamma_5/L^{1/3} \Rightarrow	f_5\leq -C_5\,,
		\end{equation}
		for all system sizes $L$.
	\end{thm}
	\begin{proof}
		We start from Eq.~\eqref{eq:f3_eq} and note 
		\begin{equation}
			\frac{[\gamma^2+(1-\gamma)^2]^{2}}{\gamma^3+(1-\gamma)^3}=1-\varepsilon_3\,,
		\end{equation}
		with $0\leq \varepsilon_3\leq \gamma$ for $0\leq \gamma\leq 1/4$. Therefore
		\begin{align}
			f_3&\leq K_3+\tilde{s}_2^2\tilde{t}_3[1-(1-\varepsilon_3)^{L-4q}]\nonumber\\
			&\leq K_3+ H_3 [1-(1-\gamma)^{L-4q}]\,.
		\end{align}
		Using $1-\gamma\geq e^{-2\gamma}$ for $0\leq \gamma\leq 1/4$, we get
		\begin{align}
			f_3&\leq K_3+ H_3 [1-e^{-2\gamma(L-4q)}]\nonumber\\
			&\leq K_3+ 2H_3 \gamma L\,.
		\end{align}
		Hence, if $\gamma\leq \gamma_3/L$, we finally arrive at
		\begin{equation}
			f_3	\leq K_3/2=-C_3\,.
		\end{equation}
		Analogously, we have
		\begin{equation}
			\frac{[\gamma^4+(1-\gamma)^4]^{2}}{[\gamma^3+(1-\gamma)^3][\gamma^5+(1-\gamma)^5]}=1-\varepsilon_5\,,
		\end{equation}
		with $0\leq \varepsilon_3\leq 2\gamma^3$ for $0\leq \gamma\leq 1/4$. Therefore
		\begin{align}
			f_5&\leq K_5+\tilde{s}_4^2\tilde{t}_5[1-(1-\varepsilon_5)^{L-4q}]\nonumber\\
			&\leq K_3+ H_5 [1-(1-2\gamma^3)^{L-4q}]\,.
		\end{align}
		Using $1-x\geq e^{-2x}$ for $0\leq x\leq 1/4$, we get
		\begin{align}
			f_5&\leq K_5+ H_5[1-e^{-4\gamma^3(L-4q)}]\nonumber\\
			&\leq K_5+ 4H_5 \gamma^3 L\,.
		\end{align}
		Hence, if $\gamma\leq \gamma_5/L^{1/3}$, we arrive at
		\begin{equation}
			f_5	\leq K_5/2=-C_5\,.
		\end{equation}
	\end{proof}
	
	\bibliography{refs}
	
\end{document}